\documentclass[sigconf,nonacm]{acmart}

\AtBeginDocument{%
  \providecommand\BibTeX{{%
    \normalfont B\kern-0.5em{\scshape i\kern-0.25em 
    b}\kern-0.8em\TeX}}}

\usepackage{xspace}
\usepackage{algorithm}
\usepackage[noend]{algorithmic}
\usepackage[labelformat=simple]{subcaption}
	
\usepackage{pgfplots}
\pgfplotsset{compat=1.3}
\usepgfplotslibrary{fillbetween}
\usetikzlibrary{positioning,calc}
\tikzset{loc/.style={draw,rectangle,rounded corners,inner sep=0.5mm},
         tn/.style={draw,rectangle},
         t/.style={->,>=stealth}}
\usepackage[inline]{enumitem}
\newenvironment{enumin}[3][.]{\begin{enumerate*}[label=\arabic*),itemjoin={{#2 }},itemjoin*={{#2 #3 }},after={#1}]}{\end{enumerate*}}
\usepackage{multirow}
\usepackage{booktabs}
\usepackage{hyperref}

\newtheorem{problem}{Problem}
\newtheorem{remark}{Remark}

\newcommand{\x}{\ensuremath{\mathbf{x}}}
\newcommand{\y}{\ensuremath{\mathbf{y}}}
\newcommand{\z}{\ensuremath{\mathbf{z}}}
\newcommand{\w}{\ensuremath{\mathbf{w}}}
\renewcommand{\a}{\ensuremath{\mathbf{a}}}
\renewcommand{\b}{\ensuremath{\mathbf{b}}}
\renewcommand{\v}{\ensuremath{\mathbf{v}}}

\newcommand{\origin}{\ensuremath{\mathbf{0}}}

\newcommand{\norm}[1]{\ensuremath{\Vert #1 \Vert}\xspace}

\newcommand{\dotp}[2]{\ensuremath{\langle #1, #2 \rangle\xspace}}

\newcommand{\expmatrix}{\textsc{ExpMatrix}\xspace}
\newcommand{\maximize}{\textsc{Max}\xspace}
\newcommand{\coordinate}{\textsc{Proj}\xspace}
\newcommand{\absv}{\textsc{Abs}\xspace}
\newcommand{\reachset}{\textsc{Reach}\xspace}
\newcommand{\ball}{\textsc{Ball}\xspace}
\newcommand{\contract}{\textsc{Contract}\xspace}
\newcommand{\False}{\texttt{False}\xspace}
\newcommand{\True}{\texttt{True}\xspace}

\newcommand{\rootn}{\textsc{RootNode}\xspace}
\newcommand{\initialize}{\textsc{InitTree}\xspace}

\newcommand {\I} {\ensuremath{\mathcal{I}}\xspace}
\newcommand {\A} {\ensuremath{\mathcal{A}}\xspace}
\newcommand {\B} {\ensuremath{\mathcal{B}}\xspace}

\renewcommand {\H} {\ensuremath{\mathcal{H}}\xspace}

\newcommand{\N}{\ensuremath{\mathbb{N}}\xspace}
\newcommand{\Npos}{\ensuremath{\mathbb{N}^+}\xspace}
\newcommand{\R}{\ensuremath{\mathbb{R}}\xspace}
\newcommand{\Rnn}{\ensuremath{\mathbb{R}_{\geq 0}}\xspace}
\newcommand {\pset} {\ensuremath{\mathcal{P}}\xspace}
\newcommand {\rset} {\ensuremath{\mathcal{R}}\xspace}
\newcommand{\eps}{\ensuremath{\varepsilon}\xspace}
\newcommand{\ebox}[2][\eps]{\ensuremath{B_{#1}( #2)}\xspace}

\newcommand{\treeset}[1]{\ensuremath{\mathbb{S}(#1)}\xspace}

\newcommand{\constraints}{{\tt{constr}}\xspace}
\newcommand{\vertices}{{\tt{vert}}\xspace}
\newcommand{\chull}{{\tt{chull}}\xspace}
\newcommand {\cpoly} {{\tt{cpoly}}\xspace}

\newcommand{\s}{{\tt{s}}\xspace}

\newcommand{\T}{\ensuremath{\mathcal{T}}\xspace}

\newcommand{\ldha}{\textsc{ldha}\xspace}
\newcommand{\adha}{\textsc{adha}\xspace}
\renewcommand {\H} {\ensuremath{\mathcal{H}}\xspace}
\newcommand {\Q} {\ensuremath{\textit{Q}}}
\newcommand {\E} {\ensuremath{\textit{E}}}
\newcommand {\X} {\ensuremath{\textit{X}}}
\newcommand {\Flow} {\ensuremath{\textit{Flow}}}
\newcommand {\Inv} {\ensuremath{\textit{Inv}}}
\renewcommand {\G} {\ensuremath{\textit{Grd}}}

\newcommand {\V} {\ensuremath{\textit{N}}}

\newcommand{\exec}{{\tt {exec}}\xspace}

\newcommand{\dom}{{\tt {dom}}\xspace}
\renewcommand{\path}{{\tt {paths}}\xspace}
\newcommand{\reach}{{\tt {SReach}}\xspace}

\newcommand{\ts}{{\tt t_s}\xspace}
\newcommand{\fend}{f_\text{end}\xspace}
\newcommand{\lreach}{{\tt {Reach}}\xspace}
\newcommand{\last}[1]{{\tt {last}}(#1)\xspace}
\newcommand{\len}[1]{{\tt {len}}(#1)\xspace}
\newcommand{\modha}{{\tt {Mod}}\xspace}

\newcommand{\children}[1]{{\tt {children}}(#1)\xspace}
\newcommand{\activated}[1]{{\tt {Act}}(#1)\xspace}

\newcommand {\upd} {\ensuremath{\textit{upd}}\xspace}

\newcommand{\follow}[2]{#1 \leadsto #2}
\newcommand{\sync}[2]{#1 \Vert #2}

\newcommand{\pwa}{\textsc{pwa}\xspace}

\newcommand{\tube}{\ensuremath{\mathbb{T}}\xspace}
\newcommand{\rest}[2]{\ensuremath{#1\!\!\downharpoonright_{#2}}}

\newcommand{\Con}{red!30!white}
\newcommand{\Coff}{blue!20!white}
\newcommand{\cI}{red!30!white}
\newcommand{\cII}{blue!20!white}
\newcommand{\cIII}{green!30!white}
\newcommand{\cIV}{yellow!20!white}

\begin{document}

\title{Synthesis of Hybrid Automata with Affine Dynamics from Time-Series Data}

\author{Miriam Garc\'ia Soto}
\affiliation{
 \institution{IST Austria}
 \city{Klosterneuburg}
 \country{Austria}}
\email{miriam.garciasoto@ist.ac.at}
\orcid{0000-0003-2936-5719}

\author{Thomas A. Henzinger}
\affiliation{
 \institution{IST Austria}
 \city{Klosterneuburg}
 \country{Austria}}
\email{tah@ist.ac.at}
\orcid{0000-0002-2985-7724}

\author{Christian Schilling}
\affiliation{
 \institution{University of Konstanz}
 \city{Konstanz}
 \country{Germany}}
\email{christian.schilling@uni-konstanz.de}
\orcid{0000-0003-3658-1065}

\begin{abstract}
Formal design of embedded and cyber-physical systems relies on 
mathematical modeling. In this paper, we consider the model class of hybrid
automata whose  dynamics are defined by affine differential equations.
Given a set of time-series data, we present an algorithmic approach to
synthesize a hybrid automaton exhibiting behavior that is close to the
data, up to a specified precision, and changes in synchrony with the data.
A fundamental problem in our synthesis algorithm is to check membership of
a time series in a hybrid automaton. Our solution integrates reachability and 
optimization techniques for affine dynamical systems to obtain both a
sufficient and a necessary condition for membership, combined in a refinement 
framework.
The algorithm processes one time series at a time and hence can be interrupted,
provide an intermediate result, and be resumed.
We report experimental 
results demonstrating the applicability of our synthesis approach.
\end{abstract}


\maketitle

\section{Introduction}

Formal design and verification of embedded control systems
require a mathematical model capturing the dynamics of each component
in the system. In general, embedded systems combine analog and digital
components. The analog components evolve continuously
in real time, while the digital components evolve in discrete time.
An appropriate mathematical formalism for modeling systems with mixed
continuous and discrete behavior is a hybrid automaton~\cite{henzinger00}.

In this paper we propose an automated approach to synthesizing a hybrid 
automaton with affine continuous dynamics (abbreviated \adha) from time-series
data in an \emph{online} fashion.
The design of models from observed data has been extensively studied in control 
theory for autoregressive 
systems~\cite{Verdult04,BemporadGPV05,PaolettiJFV07,vidal08,GarulliPV12}, 
which can be seen as discrete dynamical systems, in contrast to the  
\emph{continuous} dynamics captured by a hybrid automaton. Most of these 
approaches process a single time-series or all data at once.
In a setting where not all data is available at once, it is
desirable to have an \emph{online} approach that processes time-series data
sequentially and iteratively updates a model; only a few approaches
support this feature~\cite{SkeppstedtJM92,VidalA04,HashambhoyV05,SotoHSZ19}.

Our synthesis approach operates in two phases.
In the first phase we transform a (discrete) time-series into a
piecewise continuous trajectory $f$, for which we present
an optimization procedure that allows to specify the error between the data
and the trajectory.
The trajectories $f$ we consider are piecewise-affine (\pwa) functions where
each piece is the
solution of an affine dynamical system of the form $\dot{\x} = A\x + \b$.
\pwa trajectories can model a large class of physical processes and
approximate generic nonlinear systems.

In the second phase, which is independent of how the \emph{continuous} \pwa
trajectory $f$ has been obtained, we synthesize an \adha from~$f$.
More precisely, we construct an \adha from an existing \adha
(initialized with the ``empty'' \adha) in two 
stages: 1)~(membership) we determine whether the new trajectory is already
captured by an
execution of the model, up to a predefined precision, and 2)~(model update) if
the trajectory is
not captured, we modify the model such that, after the modification, the new
model captures the trajectory (and all trajectories that had been captured before).

We propose a three-step algorithm for the \emph{membership} problem (``is a \pwa 
trajectory captured by an \adha?'').
The first step is a reachability analysis inside a tube around the trajectory that we 
use to provide a negative answer. This problem has been studied 
in~\cite{SotoHSZ19} for the class of hybrid automata with 
piecewise-constant dynamics.
The second step is an optimization-based analysis that we use to provide a 
positive answer.
The third step is a refinement procedure to deal with cases when the first 
two steps were not conclusive.

If we find that the \pwa trajectory is not captured by the model in the 
membership query, we apply a \emph{model update} by adding 
behavior
to the automaton.
We first try to relax the continuous constraints of the automaton
(called invariants and guards).
If this relaxation is not sufficient to capture the trajectory, we also apply 
structural changes to the automaton (adding transitions and locations).

In summary, we present algorithms to solve the following problems for \pwa
trajectories and \adha{s} with a given precision:
\begin{itemize}
	\item transforming time-series data to \pwa trajectories (Section~\ref{sec:ts2function})
	
	\item membership of a \pwa trajectory in an \adha (Section~\ref{sec:membership})
	
	\item synthesizing an \adha from \pwa trajectories (Section~\ref{sec:update})
\end{itemize}
Together, our algorithms form an end-to-end approach to the
synthesis of an \adha from time-series data with a given precision.

\paragraph*{Related work}
The synthesis of hybrid systems has been explored previously
in different fields and is known as 
\emph{identification} in the area of control theory (see the surveys~\cite{PaolettiJFV07, GarulliPV12}) and as
\emph{process mining} and \emph{model learning} to a broader research
community.
Most of the techniques focus on input-output models, such as
switched autoregressive exogenous (SARX)~\cite{HashambhoyV05, Ozay16} and
(PWARX) models~\cite{Ferrari01, Ferrari03, Roll04, Nakada05, 
Juloski05, BemporadGPV05}.
SARX models constitute a subclass of linear hybrid automata (which, unlike the \adha, only has dynamics with constant derivatives)
with deterministic switching behavior and PWARX models are piecewise
ARX models where the regressor space forms a state-space polyhedral partition.
The aforementioned methods mainly consider
single-input single-output (SISO) systems, whereas a few of them
consider multiple-input multiple-output (MIMO) systems~\cite{Huang04, Verdult04,vidal08}.
Other techniques identify piecewise affine systems
in state-space form~\cite{Verdult04, Munz05, Alur14}.
The identification techniques can also be
classified into optimization-based methods~\cite{Lauer08,Ozay09}
clustering-based procedures~\cite{Ferrari03,Nakada05} and
algebraic approaches~\cite{vidal08,Nazari16}.
Most of these methods are proposed for offline identification,
with some exceptions~\cite{SkeppstedtJM92, VidalA04, HashambhoyV05}.
We propose an online approach that synthesizes hybrid automata
with affine dynamics, which are systems in state-space form.

In the field of computer science, we find techniques for learning
models from traces, which refers to approaches based on learning finite-state
machines~\cite{Angluin87} or other machine-learning techniques.
Most approaches learn a (simpler) linear hybrid automaton.
The work in~\cite{MedhatRBF15} describes an abstract framework, based on
heuristics, to learn offline from input-output traces by first learning
the discrete structure and later adding continuous dynamics.
Bartocci et al.\ learn \emph{shape expressions}, which have a similar
expressiveness~\cite{BartocciDGMNQ20}.
A recent online approach provides soundness and precision
guarantees~\cite{SotoHSZ19}.
However, that approach is restricted to linear hybrid
automata, i.e., constant dynamics.
We consider affine dynamics and follow a principled search algorithm for the
automaton modification.

We are not aware of approaches that transform time-series data to continuous
affine dynamical systems.
Some approaches consider discrete-time models, such as the work by Willems
for LTI systems~\cite{Willems86a}, and other approaches for SARX models based 
on convex optimization~\cite{OzayLS15} or generalized principal component 
analysis~\cite{vidal08}.

\section{Basic definitions}

\textit{Sets.}
Let $\R$,  $\Rnn$, and $\N$ denote the set of real numbers, non-negative real
numbers, and natural numbers, respectively.
Given a set $X$, the \emph{power set} $\pset(X)$ is the set of all subsets of $X$.
We write $\x$ for points $(x_1 , \ldots , x_n)$ in $\R^n$.
Given a point $\x \in \R^n$ and $\eps \in \Rnn$, we define the
\emph{ball} of radius \eps around $\x$ as
$\ebox{\x} := \{\y \in \R^n : \norm{\x - \y} \leq \eps\}$,
where $\norm{\cdot}$ is the infinity norm.
Given two sets $P, P' \subseteq \R^n$, we define the \emph{distance} between
$P$ and $P'$ as
$d(P,P') := \inf \{ \norm{\x - \y} : \x \in P,\linebreak[1] \y \in P' \}$.
Let $\a \in \R^n$ and $b \in \R$ be constant and $\x$
be a variable in $\R^n$, and let $\dotp{\a}{\x}$ denote the dot product of $\a$ and $\x$; then $\dotp{\a}{\x} \sim b$ is a \emph{linear constraint} 
where $\sim\ \in \{ =, \leq \}$, the set $\{\x : \dotp{\a}{\x} = b\}$ is a
\emph{hyperplane}, and the set $\{\x : \dotp{\a}{\x} \leq b\}$ is a \emph{half-space}.
A \emph{(convex) polytope} is a compact intersection of linear constraints.
Equivalently, a polytope is the convex hull of a set of vertices
$\v_1, \dots, \v_m \in \R^n$, written $\chull(\{\v_1, \dots,\v_m\})$.
For a polytope $P$ we denote the set of its linear constraints by
$\constraints(P)$ and the set of its vertices by $\vertices(P)$.
Let $\cpoly(n)$ be the set of convex polytopes over $\R^n$.

\textit{Trees.}
A \emph{tree} is a directed acyclic graph $\T = (\V,\E)$ with finite set of nodes $\V$, including a \emph{root} node, and edges
$\E \subseteq \V \times \V$.
Given a node $\nu \in \V$, the \emph{child nodes} are
$\children{\nu} = \{ \nu' \in 
\V 
: (\nu, \nu') \in \E \}$.

\textit{Functions, dynamical systems, and trajectories.}
Given a function $f$, let $\dom(f)$ denote its domain. Let $\rest{f}{D}$ denote the restriction of $f$ to domain $D \subseteq \dom(f)$. Given two functions $f$ and $g$ with $\dom(f) = \dom(g)$, the distance between $f$ and $g$ is denoted by $d(f,g)$ and defined as $\max_{t \in \dom(f)} \norm{f(t) - g(t)}$.
We typically have $\dom(f) = [0,T]$, where the \emph{initial and final states} of $f$ correspond to $f(0)$ and $f(T)$
and are denoted by $f_0$ and $\fend$, respectively.
A \emph{time series} is a sampling $s : D \to \R^n$ over a finite time domain $D$.

A function $f: [0,T] \to \R^n$ is a \emph{piecewise-affine (\pwa) trajectory} with $k$ pieces if it is continuous and there is a tuple $(\I, \A, \B)$ where $\I$ is a finite set of consecutive time intervals $[t_0, t_1], \ldots, [t_{k-1}, t_k]$ with $[0,T] = \cup_{1 \leq i \leq k} [t_{i-1}, t_i]$, $\A$ and $\B$ are $k$-tuples of matrices $A_i \in \R^{n \times n}$ and vectors $\b_i \in \R^n$, respectively, $i = 1, \dots, k$, and $\rest{f}{[t_{i-1}, t_i]}$ is a solution of the affine dynamical system $\dot{\x} = A_i\x+\b_i$, where $\dot{\x}$ denotes the derivative of $\x$ with respect to $t$.
We assume that $\pwa$ trajectories are given as the above tuple.
We call $\rest{f}{[t_{i-1}, t_i]}$ the \emph{pieces} of $f$,
and $\ts(f) := \{ t_0, \ldots, t_k \}$
the \emph{switching times} of $f$.
Each piece of $f$ is called an \emph{affine} trajectory.
A \emph{linear trajectory} $f$ is a special case of an affine trajectory where $\b = \origin$.

\subsection{Hybrid automaton with affine dynamics}

We consider a particular class of hybrid automata~\cite{henzinger00}
with invariants and guards given by linear constraints and with continuous dynamics given by affine differential equations.
\begin{definition}\label{def:ha}
	An $n$-dimensional \emph{hybrid automaton with affine dynamics}
	$(\adha)$ is a tuple $\H = (\Q, \E, \X, \Flow, \Inv, \G)$, where
	\begin{enumin}{,}{and}
		\item
		$\Q$ is a finite set of locations	
		\item
		$\E \subseteq \Q \times \Q$ is a transition relation
		\item
		$\X = \R^n$ is the continuous state space
		\item
		$\Flow: \Q \to \R^{n \times n} \times \R^n$ is the injective flow function that returns a matrix $A$ and a vector \b, and we write $\Flow_A(q) \in \R^{n \times n}$ and $\Flow_\b(q) \in \R^n$ to refer to each component
		\item
		$\Inv: \Q \to \cpoly(\R^n)$ is the invariant function
		\item
		$\G: \E \to \cpoly(\R^n)$ is the guard function
	\end{enumin}
\end{definition}

A \emph{path} $\pi$ in $\H$ of length $k$ is a sequence of locations
$q_1, \ldots, q_k$ in $\Q$ such that $(q_i, q_{i+1}) \in \E$ for each
$1 \leq i < k$. We write $\path(\H)$ for the set of paths in $\H$.
Given a path $\pi = q_1, \ldots, q_k$, we define $\len{\pi} = k$ as
the length of $\pi$  and $\last{\pi} = q_k$ as the last location.

Next we define an execution of an \adha, describing
the evolution of the continuous state subject to time passing
and discrete switches.

\begin{definition}\label{def:exec}
	An \emph{execution} $\sigma$ of an \adha $\H$ is a \pwa trajectory
	$\sigma: I \to \R^n$
	such that there is a path $\pi = q_1, \ldots, q_k$
	in $\H$ and a sequence of time points $t_0 = 0, t_1, \ldots, t_k$
	satisfying
\begin{enumin}{,}{and}
		\item
		$I = [t_0, t_k]$
		\item
		$\sigma(t) \in \Inv(q_i)$ for every $1 \leq i \leq k$ and $t \in 
		[t_{i-1}, t_i]$ 
		\item 
		$\sigma(t_i) \in$ $\G(q_i,q_{i+1})$ for every $1 \leq i < k$ 
		\item
		$\dot{\sigma}(t) = \Flow_A(q_i) \cdot \sigma(t) $ $+ 
		\Flow_\b(q_i)$ 
		for every $1 \leq i \leq k$ and $t \in (t_{i-1}, t_i)$
	\end{enumin}
\end{definition}

Thus switches between dynamics are state-dependent.
We call $\ts(\sigma) = \{t_0, \ldots, t_k\}$ the \emph{switching times} of $\sigma$.
We say that $\sigma$ \emph{follows} $\pi$, written $\follow{\sigma}{\pi}$,
and denote the set of executions by $\exec(\H)$.

\section{Problem statement}\label{sec:problem}

Our overall goal is to synthesize a hybrid automaton from data, given in the form of time series, such that the synthesized automaton captures the dynamical behavior of the data up to a given precision.
We split up this problem into two phases.
In the first phase, given a time series $s$ and a value
$\delta \in \Rnn$, we find a \pwa trajectory $f$ that is $\delta$-close
to all points in $s$.

\begin{definition}
	Given a time series $s$ with domain $D \subseteq [0,T]$,
	a \pwa trajectory $f$ with $\dom(f) = [0,T]$,
	and a value $\delta \in \Rnn$, we say that
	$f$ \emph{$\delta$-captures} $s$ if
	$\norm{s(t) - f(t)} \leq \delta$ for each $t \in D$.
\end{definition}

In the second phase, given another value $\eps \in \Rnn$, we construct
a hybrid automaton from this \pwa trajectory.

\begin{definition}
	Given a \pwa trajectory $f$ and a value $\eps \in \Rnn$, we say that
	an \adha $\H$ \emph{\eps-captures} $f$ if there exists an execution
	$\sigma \in \exec(\H)$ such that $d(f,\sigma) \leq \eps$.
\end{definition}

The definition extends to a set $F$ of piecewise-affine trajectories, i.e., $\H$ \eps-captures $F$ if $\H$ \eps-captures each $f$ in $F$.
A possible problem to consider is:
\textit{Given a set of \pwa trajectories $F$ and $\eps \in \Rnn$, construct an \adha $\H$ such that $\H$ \eps-captures $F$.}
The construction of a universal automaton, describing every possible behavior, trivially satisfies the constraint but is not a useful model. Our goal is to construct a model with a \emph{reasonable} amount of behavior by introducing a minimality criterion that we formally discuss later.

\begin{problem}[Synthesis]\label{problem:synthesis}
Given a set of \pwa trajectories $F$ and $\eps \in \Rnn$, construct an \adha $\H$ such that $\H$ \eps-captures $F$ and satisfies a minimality criterion.
\end{problem}

We propose an approach that processes one trajectory $f$ in $F$ at a time and proceeds in two stages.
Given a hybrid automaton $\H$ and a \pwa trajectory $f$, in the first stage we check whether $\H$ \eps-captures $f$, which we call a \emph{membership query}.
In the second stage, if $f$ is not \eps-captured, we modify $\H$ such that it \eps-captures $f$.
This modification may consist of several changes to the model: increasing the invariants and guards, adding new transitions, and adding new locations. We prioritize the modifications in the order given above to minimize the number of locations.

In the next three sections we present algorithmic approaches to transforming time series to \pwa trajectories, solving the membership query, and performing the model update.

\section{From time series to \pwa trajectory}\label{sec:ts2function}

In the first phase of our algorithmic framework we construct a \pwa trajectory from a time series $s$.
Recall that $f$ is supposed to be the solution of a piecewise-affine dynamical 
system, i.e., of a sequence of contiguous solutions of systems of the form 
$\dot{\x}(t) = A_i\x(t) + \b_i$ with $\x(0) = \x_0$.
We simplify the problem of finding $f$ by only considering switching times of $f$ from the domain of $s$.

We thus need to solve the following simpler problem.
Given a time series $s$ with domain $D$ and a value $\delta \in \Rnn$, find an 
affine dynamical system $\dot{\x}(t) = A\x(t) + \b$ and an initial state $\x(0) = 
\x_0$ such that the solution $g$ satisfies $\Vert s(t) - g(t) \Vert \leq \delta$ for 
every $t \in D$, or determine that no such system exists.
We pose the problem of finding $g$ as a parameter identification problem where the parameters are the coefficients of $A$, \b, and $\x_0$.
This can be written as a query to an optimization tool in combination with an ODE 
solver (we refer to Section~\ref{sec:evaluation} for implementation details).
Given concrete parameter values, i.e., instances of $A$, \b, and $\x_0$, the ODE solver can compute the solution $g$ corresponding to the affine dynamical system.
We can hence evaluate the norm $\Vert s(t) - g(t) \Vert$ at all time points $t \in D$.
The optimization tool thus has to find a solution $g$ such that this norm at those time points is less than $\delta$.

We can use the above algorithm for solving the original problem of finding a 
\pwa trajectory.
The main idea is to maximize the duration in which we can use the same dynamics.
Denote the time points of $s$ by $t_0 < \dots < t_k$.
We first find the maximum time point $t_i$ such that the above-described algorithm finds a solution (e.g., using binary search).
Then we iteratively solve the same problem for the time-series suffix from $t_i$ to $t_k$, until finally $t_i = t_k$.
Note that we only need to identify $\x_0$ for the first piece, as for subsequent pieces the initial state is determined by $\x_0$ and the previous dynamics.

\section{Membership query}\label{sec:membership}

In this section we formalize and solve the membership query. Given an \adha $\H$, a \pwa trajectory $f$, and a value $\eps \in \Rnn$, the fundamental problem we need to solve is to determine if $\H$ \eps-captures $f$.
We reduce this problem to checking whether for a given \pwa trajectory $f$ and a given path $\pi$ in $\H$ there exists an execution $\sigma$ following $\pi$ such that $d(f,\sigma) \leq \eps$.
We apply this check to every path $\pi$ in $\H$ of length equal to the number of pieces in $f$. We provide a solution by restricting $f$ and $\sigma$ to switch synchronously,
which allows us to evaluate the pieces consecutively.

\begin{definition}
	An execution $\sigma$ of an \adha $\H$ is \emph{synchronized} with a \pwa trajectory $f$, denoted by $\sync{\sigma}{f}$, if $\dom(\sigma) = \dom(f)$ and $\ts(\sigma) =\ts(f)$.
\end{definition}

\begin{problem}[Membership]\label{pb:syncmem}
	Given a path $\pi$ in an \adha $\H$, a \pwa trajectory $f$, and $\eps \in \Rnn$, determine if there exists a synchronized execution $\sigma$ of $\H$ with $\follow{\sigma}{\pi}$ and $d(f,\sigma) \leq \eps$.
\end{problem}

Our membership algorithm uses reachability analysis to approximate the states that the synchronized executions of $\H$ can reach.

\begin{definition}
	Given an $n$-dimensional \pwa trajectory $f$ and $\eps \in \Rnn$, an \emph{\eps-tube} of $f$ is the function $\tube(f,\eps): \Rnn \to \pset(\R^n)$ such that $\tube(f,\eps)(t) = \ebox{f(t)}$.
\end{definition}

\begin{definition}
	Given an \adha $\H$, a path $\pi \in \path(\H)$, a \pwa trajectory $f$, and $\eps \in \Rnn$, the \emph{synchronized reachable set}, starting from a set $P \subseteq \R^n$ and following $\pi$, is defined as
	\begin{align*}
	\reach(P,\pi,f,\eps) := \{ \x \in \R^n : \exists \sigma \in \exec(\H), 
	\follow{\sigma}{\pi}, \\\sync{\sigma}{f},
	\sigma_0 \in P, \sigma(t) \in \tube(f,\eps)(t) \, \forall t \in \dom(\sigma) \text{ and } \sigma_{end} = \x \}.
	\end{align*}
\end{definition}

For Problem~\ref{pb:syncmem}, an execution in $\H$ satisfying the corresponding constraints exists if $\reach(P,\pi,f,\eps)$ is nonempty.
Note that the converse is not true due to unsynchronized executions.

\begin{proposition}
	Let $\H$ be an \adha, $f$ be a \pwa trajectory, $P \subseteq \R^n$, and $\eps \in \Rnn$. If $\reach(P,\pi,f,\eps)$ is nonempty for some $\pi \in \path(\H)$, then $\H$ \eps-captures $f$.
\end{proposition}

We inductively construct the synchronized reachable set for a \pwa trajectory $f$ by computing the synchronized reachable set for each affine piece of $f$. Concretely, given an initial set $P$, a path $\pi = q_1, \ldots, q_k$ in $\H$, and a \pwa trajectory $f$ with $\ts(f) = t_0, \ldots, t_k$,
we define the synchronized reachable sets
\begin{equation}\label{eq:syncset}
P_0 := P, \hspace{2mm}
P_i := \reach(P_{i-1}, q_i, \rest{f}{[t_{i-1}, t_i]},\eps) \text{ for } 1 \leq i 
\leq k.
\end{equation} 
Observe that $P_k$ is equal to $\reach(P,\pi,f,\eps)$.

\subsection{Membership query for single trajectories}

We now present a method to approximate the synchronized reachable set for a \pwa trajectory $f$ with just one piece, starting from a polytope $P$ and following a path $q$ of length one in $\H$, that is, $\reach(P,q,f,\eps)$.
This is a special case of Problem~\ref{pb:syncmem} where $f$ is an affine trajectory and the path $\pi$ in $\H$ is a single location $q$.
As observed before, checking emptiness of the synchronized reachable set is equivalent to checking whether there exists of an affine trajectory $\sigma$ in the \eps-tube of $f$, starting from the given polytope $P$, with the same time domain as $f$, and following the dynamics of~$q$.

\begin{remark}
	Without loss of generality we restrict ourselves to linear dynamics,
	which are equivalent to affine dynamics under an appropriate
	transformation: Add an extra variable $y$ to an affine system $\dot{\x} = 
	A\x+\b$ as $\dot{\x} = A\x+\b y$ where $y$ 
	is constant $1$ (i.e., $\dot{y} = 0$).
	Hence we also consider hybrid automata with linear dynamics (\ldha), which 
	means that the flow function has the signature $\Flow: \Q \to \R^{n \times n}$.
\end{remark}

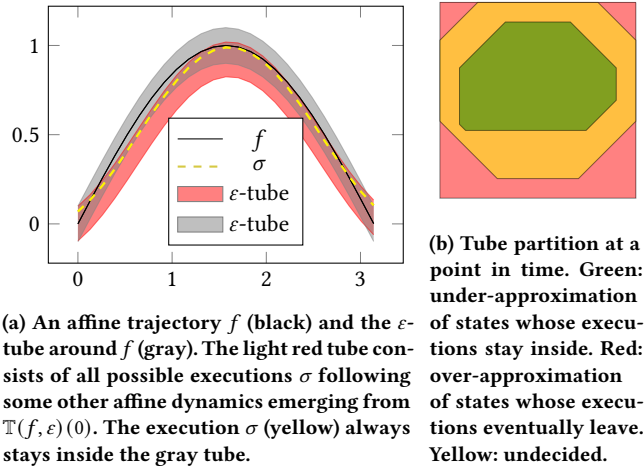
\begin{figure}
	\centering
	\begin{subfigure}[b]{.64\linewidth}
		\centering
		\begin{tikzpicture}
	\begin{axis}[width=63mm,height=50mm,domain=0:3.14,legend style={matrix anchor=west,at={(12mm,8mm)}}]
		\addplot[black] {sin(deg(x))};
			\addlegendentry{$f$}
		\addplot[yellow!80!black,thick,dashed] {0.5*sin(deg(1.6*x-1)) + 0.49};
			\addlegendentry{$\sigma$}
		\addplot[red,opacity=0.5,name path=lower_bottom,forget plot] {0.5*sin(deg(1.6*x-1)) + 0.325};
		\addplot[red,opacity=0.5,name path=upper_top,forget plot] {0.5*sin(deg(1.6*x-1)) + 0.52};
		\addplot[fill=red,opacity=0.5] fill between[
			of = lower_bottom and upper_top
		];
			\addlegendentry{\eps-tube}
		\addplot[gray,opacity=0.5,name path=lower,forget plot] {sin(deg(x)) - 0.1};
		\addplot[gray,opacity=0.5,name path=upper,forget plot] {sin(deg(x)) + 0.1};
		\addplot[fill=gray,opacity=0.5] fill between[
			of = lower and upper
		];
			\addlegendentry{\eps-tube}
		\addplot[black] {sin(deg(x))};
		\addplot[yellow,thick,dashed] {0.5*sin(deg(1.6*x-1)) + 0.49};
	\end{axis}
\end{tikzpicture}
		\caption{An affine trajectory $f$ (black) and the
		\eps-tube around $f$ (gray).
		The light red tube consists of all possible executions $\sigma$
		following some other affine dynamics
		emerging from $\tube(f,\eps)(0)$.
		The execution $\sigma$ (yellow) always stays inside the gray tube.}
		\label{fig:tube_example}
	\end{subfigure}
	\hfill
	\begin{subfigure}[b]{.33\linewidth}
		\centering
		\includegraphics[width=\textwidth,keepaspectratio,clip,trim=4cm 0mm 4cm 0mm]{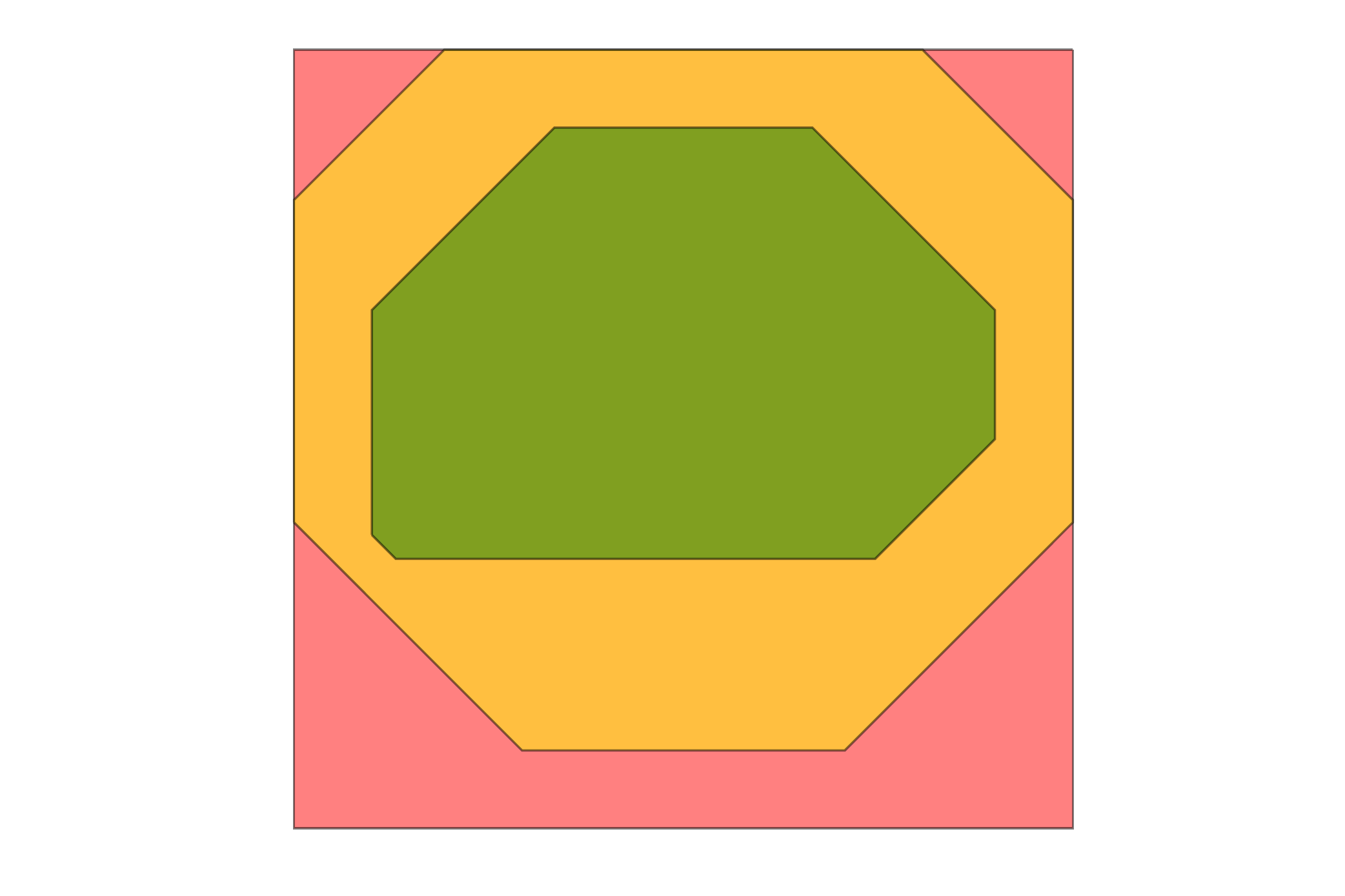}
		\caption{Tube partition at a point in time.
		Green: under-approximation of states whose executions stay inside.
		Red: over-approximation of states whose executions eventually leave.
		Yellow: undecided.}
		\label{fig:sync_reach_set}
	\end{subfigure}
	\caption{Illustration of reachability computations.}
\end{figure}

Figure~\ref{fig:tube_example} illustrates that computing the exact
synchronized reachable set is not trivial.
Hence we settle for an approximate solution by successive
polytope refinements into three regions,
corresponding to the respective executions emerging from those regions,
as illustrated in Figure~\ref{fig:sync_reach_set}:
an under-approximation of the states in $P$ whose executions definitely
stay inside the tube,
an over-approximation of the states in $P$ whose executions definitely
leave the tube,
and the remaining states that are undetermined.
In summary, we want to achieve the following goals:
\begin{enumerate}[label=(G\arabic*),leftmargin=3.4\parindent]
	\item \label{itm:g1}
	determine whether $\reach(P,q,f,\eps)$ is empty,
	\item \label{itm:g2}
	(approximately) compute $\reach(P,q,f,\eps)$, and
	\item \label{itm:g3}
	refine the polytope $P$ to improve the approximation.
\end{enumerate}
We next discuss in detail how to achieve these goals.

\subsection{Emptiness of SReach}\label{sec:emptiness}

We now work toward an algorithm for achieving goal~\ref{itm:g1}.
As argued before, solving the emptiness problem exactly
is not trivial.
A sufficient condition is to compute an \emph{over-approximation} and show
emptiness for that set.
$\reach(P,q,f,\eps)$ is empty if and only if there exists
a time point $t \in [0,T]$ such that $\reach(P, q, {\rest{f}{[0,t]}}, \eps)$
is empty.
We can generalize this observation to sets of points $P' \subseteq P$.
Observe that
$\reach(P,q,f,\eps) = \reach(P',q,f,\eps) \cup \reach(P \setminus P', q, f,\eps)$,
so if for each point $\x$ in $P'$ there exists a time point $t$
such that the execution emerging from $\x$ leaves the tube,
we can remove the set $P'$ from $P$.
We recall a classic result.

\begin{definition}
	The \emph{reachable region} from $P \subseteq \R^n$ following the linear
	dynamics described by $A \in \R^{n \times n} $ at time $t$ is defined as
	$\lreach(P,A,t) := \{ e^{At} \cdot \x : \x \in P \}$.
\end{definition}

With $A = \Flow(q) \in \R^{n \times n}$, we know that $\lreach(P,A,t)$
includes the points of all executions $\sigma$ at time $t$ such that
$\follow{\sigma}{q}$ starting from $\x_0 \in P$.
Moreover, $\sigma(t)$ belongs to the \eps-tube around $f$ at time $t$.
Therefore, $\lreach(P,A,t) \cap \tube(f,\eps)(t)$
is an \emph{over-approximation} of $\reach(P,q,\rest{f}{[0,t]},\eps)$, providing a sufficient emptiness check.

\begin{proposition}\label{prop:necessary}
	Emptiness of $\lreach(P,A,t) \cap \tube(f,\eps)(t)$ implies
	emptiness of $\reach(P,q,{\rest{f}{[0,t]}},\eps)$, which
	implies emptiness of $\reach(P,q,f,\eps)$.
\end{proposition}

Proposition~\ref{prop:necessary} suggests an algorithm for showing emptiness of
$\lreach(P,A,t) \cap \tube(f,\eps)(t)$ at sampled time points $t \in [0,T]$.
Observe that a finer time sampling provides a more accurate approximation,
and a better chance to show emptiness if $\reach(P,q,f,\eps) = \emptyset$.
For a uniform sequence of $m$ time points of delay $\delta$, 
Algorithm~\ref{alg:linear_necessary} performs the above sufficient check 
numerically.
Recall that $A$ is the flow of location $q$, $\Flow(q)$, and the linear trajectory 
$f$ is given as the tuple $(\{ [0,T] \}, \{ B \})$.
Algorithm~\ref{alg:linear_necessary} takes as input two matrices $A,B \in \R^{n 
\times n}$, a point $\x_0 \in \R^n$, a polytope $P_0 \subseteq \R^n$, two 
values $\eps, T \in \Rnn$, and a natural number $m > 0$.
For the $j$-th time step $\delta = T/m$, the algorithm computes $e^{Bj\delta} 
\cdot \x_0$, where $e^{Bj\delta}$ is obtained with the function 
$\expmatrix(Bj\delta)$.
Then $\ball(\x,\eps)$ constructs the ball $\ebox{\x}$, and $\reachset(P_{j-1}, A, 
\delta)$ computes the set $\lreach(P_{j-1}, A, \delta)$, which is intersected with 
the ball for constructing $P_j$.

\begin{algorithm}[t]
	\caption{Over-approximation of $\reach$}
	\label{alg:linear_necessary}
	\begin{algorithmic}[1]
		\REQUIRE Matrices $A,B \in \R^{n \times n}$, a point $\x_0 \in \R^n$,
		a polytope $P_0$, values $\eps, T \in \Rnn$, and a number $m \in \Npos$.
		
		\STATE $\delta$ := $T$ / $m$  \COMMENT{uniform time step for sampling}
		
		\FOR {$j \in [1, \dots, m]$}
		
			\STATE $\x$ := $\expmatrix(B j \delta) \cdot \x_0$
			\STATE $B_j$ := $\ball(\x,\eps)$  \COMMENT{tube at time point $j\delta$}
			
			\STATE $P_j$ := $\reachset(P_{j-1}, A, \delta) \cap B_j$
			
			\IF{isempty($P_j$)}  \label{line:emptiness_check}
			
				\RETURN{$P_j$}
			
			\ENDIF
		
		\ENDFOR
		
		\RETURN{$P_m$}
	\end{algorithmic}
\end{algorithm}

\begin{proposition}[Soundness]
	Algorithm~\ref{alg:linear_necessary} returns an empty set only if
	$\reach(P,q,f,\eps)$ is empty.
\end{proposition}

\begin{proof}
	Assume that the algorithm returns an empty set but
	$\reach(P,q,f,\eps)$ is nonempty.
	Then there is a point $\x \in P$ with
	$e^{At} \cdot \x \in \tube(f,\eps)(t)$ for every $t = j \delta$, $1 \leq j \leq m$.
	Hence $e^{A j \delta} \cdot \x \in P_j$,
	which contradicts the condition in line~\ref{line:emptiness_check}.
\end{proof}

\begin{proposition}[Robust completeness]~\label{prop:robust}
	Let $P$ be a polytope in $\R^n$, $A \in \R^{n \times n}$, $f$ a linear 
	trajectory with $\dom(f) = [0,T]$, and $\eps_0 > 0$ such that for every $\x \in 
	P$ there exists $t \in [0,T]$ with $d(\lreach(\{\x\},A,t), \tube(f,\eps)(t)) > 
	\eps_0$.
	Then there exists a finite number $m$ such that 
	Algorithm~\ref{alg:linear_necessary} returns an empty set.
\end{proposition}

\begin{proof}
	Fix $\x \in P$ and $t \in [0,T]$ such that
	$d(\lreach(\{\x\},A,t),$ $ \tube(f,\eps)(t)) > \eps_0$.
	Then, by continuity of the distance function,
	there exists a time $t^{exit} \in [0,t]$ such that
	$d(\lreach(\{\x\},A,t^{exit}),$ $ \tube(f,\eps)(t^{exit})) = 0$ and
	for every $t' \in [t^{exit},t]$, $d(\lreach(\{\x\},A,t'),$ $ 
	\tube(f,\eps)(t')) > 0$.
	Let us denote $t-t^{exit} = \delta_{\x}$.
	Compute the infimum of $\delta_{\x}$ for every $\x \in P$,
	denoted as $\delta^*$.
	Then, choose $m > T/\delta^*$.	
\end{proof}

\begin{remark}
	The assumption on Proposition~\ref{prop:robust} about $\eps_0$ is necessary
	in general because $P$ and the \eps-tube image are compact.
	Since $P$ and $P' := \{ \x \in P: \lreach(\{\x\},A,t) \subseteq \tube(f,\eps)(t) \, \forall t\}$
	are topologically closed, $P' \setminus P$ is not topologically closed.
\end{remark}

Algorithm~\ref{alg:linear_necessary} is a sufficient check: the result is empty only if
$\reach(P,q,f,\eps)$ is empty.
Next we consider membership of
$\sigma$ in the \eps-tube of $f$ where $\sigma$ starts from
a fixed point $\x$ in $P$.

\subsection{Approximation of SReach}\label{sec:point_reach_set}

We can achieve goal~\ref{itm:g2} (and hence goal~\ref{itm:g1})
for a singleton set $P = \{\x\}$.
In other words, for a fixed starting point $\sigma(0) = \x$ we can decide if
$\sigma(t) \in \tube(f,\eps)(t)$ for every $t \in [0, T]$.
We consider the case where $\x \in \tube(f,\eps)(0)$.
We can easily determine if $\sigma(T) \in \tube(f,\eps)(T)$
(e.g., by executing Algorithm~\ref{alg:linear_necessary} with $m = 1$).
In the nontrivial case that $\sigma(T) \in \tube(f,\eps)(T)$,
the goal is to compute the maximum of
$d(\sigma(t), \tube(f,\eps)(t))$
over time interval $[0, T]$.
Our approach to that problem involves solving $2n$ optimization problems.

\begin{proposition}[Theorem~4 in~\cite{Hainry08}]
	Let $\x \in \R^n$ be a point and $A \in \R^{n \times n}$
	a matrix with rational coefficients.
	Then $\lreach(\{ \x \},A,t)$ is computable for every time $t$.
\end{proposition}

\begin{algorithm}[t]
	\caption{Synchronization check for linear trajectories}
	\label{alg:linear_singleton}
	\begin{algorithmic}[1]
		\REQUIRE Two matrices $A, B \in \R^{n \times n}$, states $\x_0, \y_0 \in \R^n$, and values $\eps, T \in \Rnn$.
		
		\STATE $\sigma(t)$ := $\expmatrix(At) \cdot \y_0$
		\STATE $f(t)$ := $\expmatrix(Bt) \cdot \x_0$
				
		\STATE $h(t)$ := $f(t) - \sigma(t)$
		
		\IF {$ \Vert \x_0 - \y_0 \Vert \leq \eps$ \AND $ \Vert f(T) - \sigma(T) \Vert \leq \eps$}
			\STATE $v$ := $0$
			\FOR {$1 \leq i \leq n$}
				\STATE $v_{max}$ := $\maximize( \absv( \coordinate(h(t),i))), [0,T])$
				\STATE $v$ := max$(v_{max}, v)$	 
			\ENDFOR
			\IF {$v \leq \eps$}
				\RETURN{\True}
			\ENDIF
		\ENDIF
		\RETURN{\False}
	\end{algorithmic}
\end{algorithm}

We summarize the procedure in Algorithm~\ref{alg:linear_singleton}.
The inputs are two matrices $A, B \in \R^{n \times n}$,
states $\x_0, \y_0 \in \R^n$, and values $\eps, T \in \Rnn$.
Recall that $A$ is the flow of location $q$, $\Flow(q)$, and the linear trajectory 
$f$ is given as the tuple $(\{ [0,T] \}, \{ B \})$.
Initially, the algorithm defines the linear trajectories
$\sigma(t)$ and $f(t)$ and their difference $h(t)$.
If the norm of this difference is less than \eps for $t=0$
and $t=T$, the algorithm computes
the maximum (\maximize function) over $[0,T]$
of the absolute values (\absv function)
for each coordinate $i$ of $h(t)$, that is, $\coordinate(h(t),i)$.
The algorithm returns \True if
the maximum distance between $\sigma$ and $f$ is less than \eps,
and \False otherwise.
Thus the algorithm determines emptiness of
$\reach(\{\x\},q,f,\eps)$ for $\x \in P$.

\begin{proposition}\label{prop:equiv}
	Algorithm~\ref{alg:linear_singleton} returns \False if and only if\\
	$\reach(\{\x\},q,f,\eps)$ is empty.
\end{proposition}

We assume a numerically sound optimization tool in practice.
Algorithm~\ref{alg:linear_singleton} gives us a way to
obtain an under-approximation of $\reach(P,q,f,\eps)$:
apply Algorithm~\ref{alg:linear_singleton} to every vertex of $P$
and construct the convex hull of the vertices for which
Algorithm~\ref{alg:linear_singleton} returns \True.
Next we prove that this set is contained in $\reach(P,q,f,\eps)$.

\begin{proposition}\label{prop:convex}
	Let $P$ be a convex polytope, $A = \Flow(q)$, $f$ be a linear trajectory with 
	domain $[0,T]$, and $\eps$ be a value in $\Rnn$.
	Then, $\reach(P,q,f,\eps) = \lreach(P,A,T)$ if $\reach(\{\v\},q,f,\eps)$ is not 
	empty for every $\v \in \vertices(P)$.
\end{proposition}
\begin{proof}
	The inclusion $\reach(P,q,f,\eps) \subseteq \lreach(P,A,T)$ is obvious.
	Let $\reach(\v,q,f,\eps) \neq \emptyset$ for every $\v \in \vertices(P)$.
	We want to show that for every point $\x \in P$,
	$\lreach(\{\x\},A,t)$ belongs to $\tube(f,\eps)(t)$ for all $t \in [0,T]$.
	Assume there exist $\x \in P$ and
	$t \in [0,T]$ with $\lreach(\{\x\},A,t) \not\subseteq \tube(f,\eps)(t)$,
	so $\reach(\{\x\},q,f,\eps) = \emptyset$.
	We know that $\lreach(\{\x\},A,t) \subseteq \lreach(P,A,t)$.
	So $\lreach(P,A,t)$ $ \not\subseteq \tube(f,\eps)(t)$.
	Moreover, $\tube(f,\eps)(t)$ is convex for each $t \in [0,T]$.
	For any polytope $P$ and convex set $C$ it holds that
	$P \subseteq C$ if and only if $\vertices(P) \subseteq C$.
	Therefore, $\lreach(P,A,t) \subseteq \tube(f,\eps)(t)$ if and
	only if $\lreach(\vertices(P), A, t) \subseteq \tube(f,\eps)(t)$,
	i.e., $\lreach(\{\v\},A,t) \subseteq \tube(f,\eps)(t)$ for each
	$\v \in \vertices(P)$ and $t \in [0,T]$.
	By assumption, $\reach(\{\v\},q,f,\eps) \neq \emptyset$
	for each $\v \in \vertices(P)$.
	Using Proposition~\ref{prop:equiv},
	$\lreach(\{\v\},A,t) \subseteq \tube(f,\eps)(t)$ for each $\v \in \vertices(P)$.
	Hence $\lreach(P,A,t) \subseteq \tube(f,\eps)(t)$ for each $\x \in P$:
	a contradiction.
\end{proof}

\begin{corollary}
	If Algorithm~\ref{alg:linear_singleton} returns \True for all vertices
	$\v \in \vertices(P)$, then $\reach(P,q,f,\eps) = \lreach(P,A,T)$.
\end{corollary}

\subsection{Polytope refinement}\label{sec:refinement}

Recall that $P$ is a polytope, $f$ is a linear trajectory 
$f(t) = e^{Bt} \cdot \x_0$ with time domain $[0,T]$, $q$ is a location in
some \ldha $\H$ with $\Flow(q) = A$,
$\eps$ is a value in $\Rnn$, and $m > 0$ is a natural number.
We can use Algorithm~\ref{alg:linear_necessary} from Section~\ref{sec:emptiness}
to obtain an over-approximation $P$ of the synchronized reachable set.
If $P$ is nonempty, we can use Algorithm~\ref{alg:linear_singleton}
from Section~\ref{sec:point_reach_set} for every vertex of $P$, and if the
algorithm returns \True for some vertex, we have a nonempty
under-approximation and can conclude membership of $f$ in $\H$.
If Algorithm~\ref{alg:linear_singleton} returns
\False for all vertices, we cannot conclude.

Next we propose a new procedure, which together with
Proposition~\ref{prop:convex} suggests an algorithm for computing
a more precise under-approximation of $\reach(P, \pi, f, \eps)$.
Finally, these procedures together induce an algorithm to refine the over- and
under-approximations.
Intuitively, recalling Figure~\ref{fig:sync_reach_set}, this refinement narrows
the discrepancy between the the over-approximation (yellow) and the under-approximation (green).

First we observe that the over-approximation $P$ is a convex polytope.
The idea is to contract this polytope to a new polytope.
Given a value $\delta \in \Rnn$, we define the $\delta$-contraction of $P$ as follows.

\begin{definition}
	Let $P$ be a polytope, $\delta$ be a value in $\Rnn$, and
	$\constraints(P) = \{ \a_1 \x \sim b_1, \ldots, \a_m \x \sim b_m \}$.
	The \emph{$\delta$-contraction} of $P$
	is the polytope
	$P_\delta := \{\x \in \R^n : \a_1 \x \sim c_1, \ldots, \a_m \x \sim c_m \}$
	where $c_j = b_j$ if $\sim$ is '$=$' and
	$c_j = b_j - \dfrac{\delta}{\norm{a_j}_2}$ if $\sim$ is '$\leq$',
	for every $1 \leq j \leq m$.
\end{definition}

\begin{algorithm}[t]
	\caption{Polytope refinement}
	\label{alg:refine}
	\begin{algorithmic}[1]
		\REQUIRE A polytope $P$,
		two matrices $A, B \in \R^{n \times n}$,
		a state $\x_0 \in \R^n$,
		and three values $\eps, T, \delta \in \Rnn$.
		
		\STATE $V^+$ := $\emptyset$
		
		\WHILE{\True}
			
			\STATE $V^-$ := $\emptyset$
			
			\FOR {$\v \in \vertices(P)$}
			
				\IF{Algorithm~\ref{alg:linear_singleton}($A, B, \x_0, \v, \eps, T$)}
				
					\STATE $V^+$ := $V^+ \cup \{\v\}$
				
				\ELSE
				
					\STATE $V^-$ := $V^- \cup \{\v\}$
				
				\ENDIF
			\ENDFOR
			
			\IF{$V^- = \emptyset$}
				
				\RETURN $\chull(V^+)$
				
			\ENDIF
			
			\STATE $P$ := $\contract(P, \delta)$
			
		\ENDWHILE
	\end{algorithmic}
\end{algorithm}

We can hence take the over-approximation $P$, compute the
$\delta$-contraction $P'$, and then apply
Algorithm~\ref{alg:linear_singleton} to all vertices of $P$ and $P'$.
Ultimately we may have to repeat this contraction several times
(at most $\lceil d / \delta \rceil$ times, where $d$ is the diameter of $P$).
In the end, since we know that the true synchronized reachable set is convex,
we can take the convex hull of all those vertices for which
Algorithm~\ref{alg:linear_singleton} returned \True (i.e., these
vertices belong to the synchronized reachable set).
We summarize the refinement in Algorithm~\ref{alg:refine}, where the
procedure \contract applies a $\delta$-contraction.

In principle, now that we have two polytopes $P$ and $P'$
over-approximating and under-approximating the synchronized
reachable set, respectively, a natural additional refinement
procedure can be conceived where one iteratively tries to
enlarge the under-approximation or shrink the over-approximation.
We did not investigate this direction because the above scheme is already very precise in practice.
(In fact, we rather observed that the approximations become too precise; see the further discussion in Section~\ref{sec:implementation}.)

\subsection{Summary}

Algorithm~\ref{alg:single_piece} summarizes the overall procedure for computing
both an under-approximation and an over-approximation of the synchronized
reachable set for a linear trajectory $f$.
We first use Algorithm~\ref{alg:linear_necessary} to compute the
over-approximation $\overline{P_1}$.
If the over-approximation is empty, we can conclude that $f$ is not
\eps-captured.
Otherwise, taking the end state $\x_1$ of $f$ and inverting the dynamics ($\dot{f}_\text{inv}(\x) = -\dot{f}(\x)$), we
use Algorithm~\ref{alg:refine} to compute the under-approximation
$\underline{P_1}$.

\begin{algorithm}[t]
	\caption{Membership query for a single piece}
	\label{alg:single_piece}
	\begin{algorithmic}[1]
		\REQUIRE A polytope $P_0$,
		two matrices $A, B \in \R^{n \times n}$,
		a state $\x_0 \in \R^n$,
		three values $\eps, T, \delta \in \Rnn$,
		and a value $m \in \Npos$.
		
		\STATE $\overline{P_1}$ := Algorithm~\ref{alg:linear_necessary}($A, B, \x_0, P_0, \eps, T, m$)
		
		\IF{isempty($\overline{P_1}$)}
		
			\RETURN{$\emptyset$, $\emptyset$}
		
		\ENDIF
		
		\STATE $\x_1$ := $\expmatrix(AT) \cdot \x_0$
		
		\STATE $\underline{P_1}$ := Algorithm~\ref{alg:refine}($\overline{P_1}$, $-A$, $-B$, $\x_1$, \eps, $T$, $\delta$)
		
		\RETURN{$\underline{P_1}$, $\overline{P_1}$}
	\end{algorithmic}
\end{algorithm}

We illustrate the algorithm and the generalization to
multiple pieces with the following parametric
linear trajectories:
\begin{align}
	\dot{\x} &= \begin{pmatrix} 0 & 1 \\ -1 & 0 \end{pmatrix} \x, & 
\x(0) &= \begin{pmatrix} 1 \\ 1 \end{pmatrix} \label{eq:system1} 
\\
	\dot{\y} &= \begin{pmatrix} 0 & 1 - \alpha \\ -1 & 0 \end{pmatrix} 
\y, & \y(0) &= \y_0 \label{eq:system2}
\end{align}

\begin{figure}
	\centering
	\begin{subfigure}[b]{.49\linewidth}
		\centering
		\includegraphics[width=\linewidth,height=50mm,keepaspectratio]{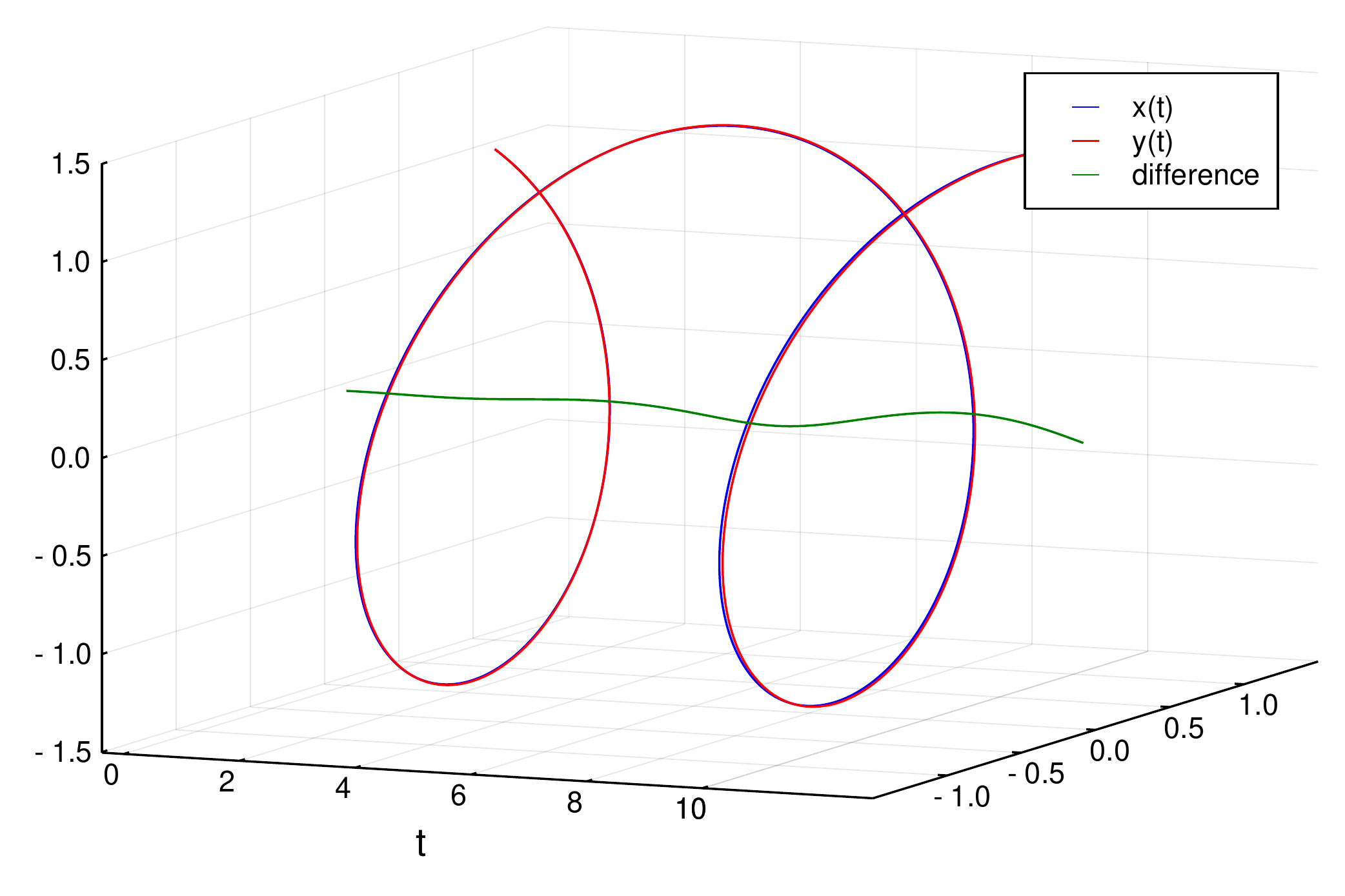}
		\caption{Two executions of system~\eqref{eq:system1} (blue)
		and system~\eqref{eq:system2} (red) and their difference (green).}
		\label{fig:execution_example_3d}
	\end{subfigure}
	\hfill
	\begin{subfigure}[b]{.47\linewidth}
		\centering
		\includegraphics[width=\linewidth,height=50mm,keepaspectratio]{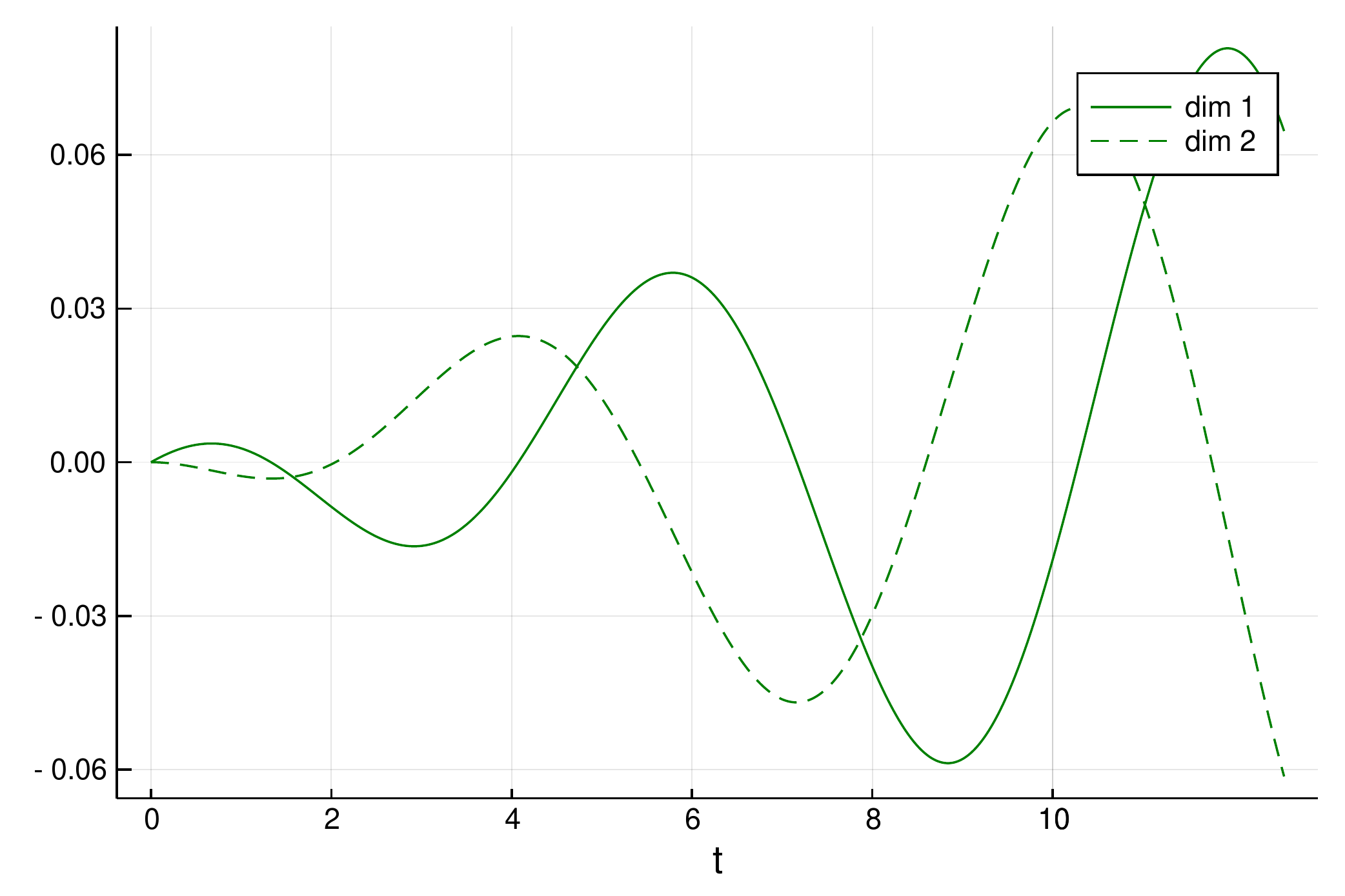}
		\caption{Difference of the trajectories projected to time and one dimension.}
		\label{fig:execution_example_diff}
	\end{subfigure}
	\caption{Trajectories and difference of systems~\eqref{eq:system1}
	and~\eqref{eq:system2} with parameters $\alpha = 0.01$ and
	$\y_0 = \x(0)$.}
	\label{fig:executions}
\end{figure}

System~\eqref{eq:system1} is fixed and takes the role of the linear trajectory $f$ 
while system~\eqref{eq:system2} models the location of an \ldha.
In the following, we fix the parameter value $\alpha$ and ask whether there
exists an initial state $\y_0$ such that the corresponding execution of
system~\eqref{eq:system2} is synchronized with $f$.
In Figure~\ref{fig:executions} we plot the executions for $\alpha = 0.01$
and the same initial state $\y_0 = \x(0)$.
It can be seen that for a time horizon of $4 \pi$ we need to choose \eps larger than ${\sim}0.08$.

In Figure~\ref{fig:membership_example} we plot the results for
$\eps = 0.1$ and $\alpha = 0$ or $\alpha = 0.01$, respectively.
In Figure~\ref{fig:membership_example_same_1} we see the
under-approximation computed by the algorithm in light green.
The dark green set is a simplified under-approximation that we use to handle
complexity, further described in Section~\ref{sec:implementation}.
Similarly, the dark yellow set is the over-approximation computed by the
algorithm, while the light yellow set is a simplified over-approximation.
It can be seen that the gap between the over-approximation (dark yellow)
and the under-approximation (light green) is very narrow, indicating that the
refinement procedure (Algorithm~\ref{alg:refine}) is precise.
Also note that the true synchronized reachable set in this case is a Euclidean
ball because, while the executions all follow the same dynamics as $f$,
those executions starting from a state outside this ball rotate around $f$
and eventually leave the tube (since the tube does not rotate).

In Figure~\ref{fig:membership_example_same_23} we plot the intermediate
sets for the same executions but modeled as \pwa trajectories with $23$ pieces,
starting with the set at time~$0$.
In theory, the settings with a single piece and $23$ pieces of the same
dynamics are equivalent; however, due to the simplifications of the
approximations for each piece, the approximations lose precision in the latter
case.
Still, the approximations in the last piece are sufficiently precise to prove
that the under-approximation (green set) is nonempty and hence we can conclude
with a positive answer to the membership query.
In the last subplot we depict a random sampling from the over-approximation
where we apply Algorithm~\ref{alg:linear_singleton} to check whether the state
indeed corresponds to a synchronized execution (green dot) or not (red dot).
Figure~\ref{fig:membership_example_different_1} shows the setting for
$\alpha = 0.01$ with similar results.

\begin{figure}
	\centering
	\begin{subfigure}[b]{0.49\linewidth}
		\centering
		\includegraphics[width=\linewidth,height=6cm,keepaspectratio]{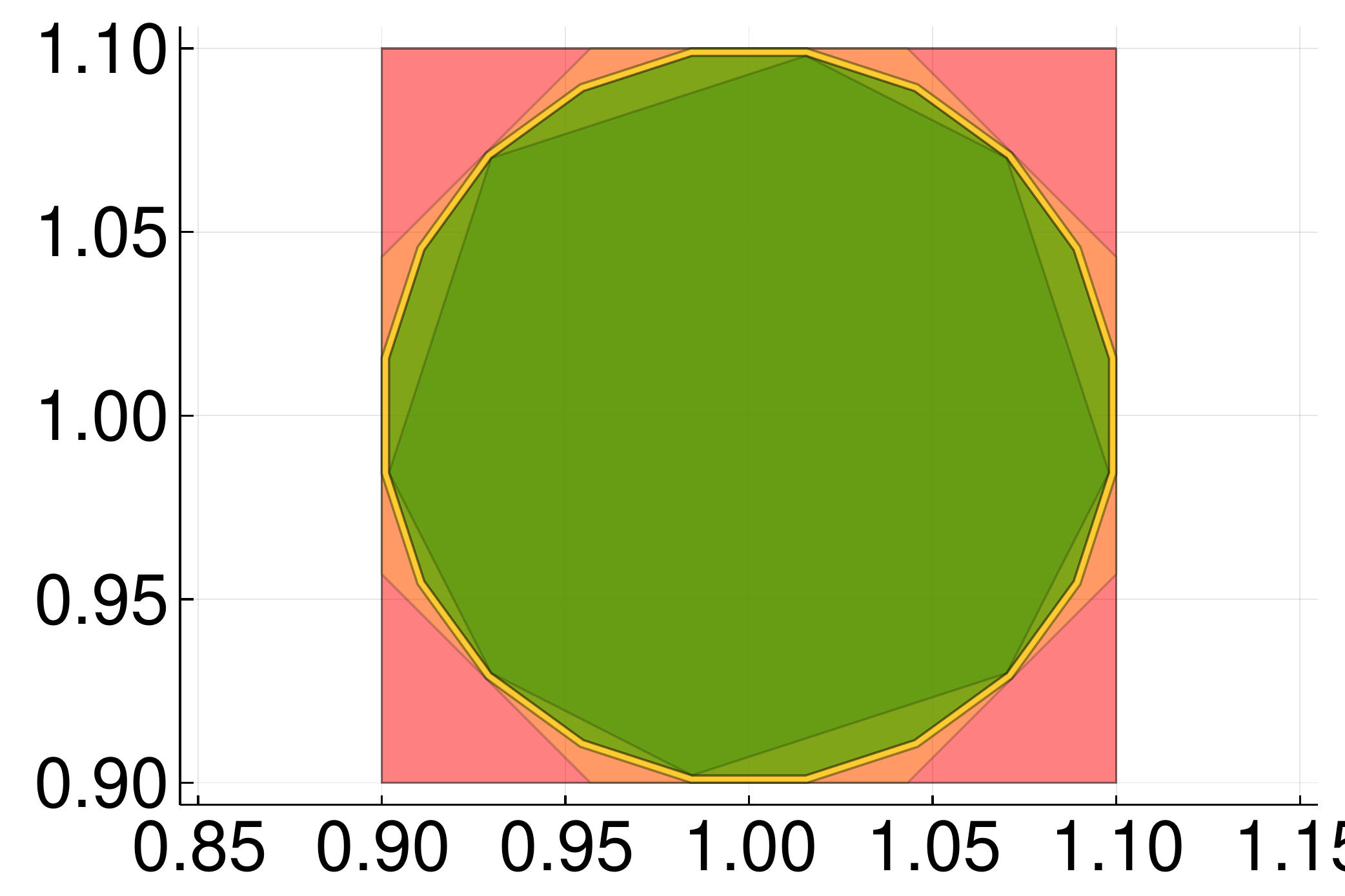}
		\caption{Analysis for $\alpha = 0$ and a single pieces of duration $4\pi$.}
		\label{fig:membership_example_same_1}
	\end{subfigure}
	\hfill
	\begin{subfigure}[b]{0.47\linewidth}
		\centering
		\includegraphics[width=\linewidth,height=6cm,keepaspectratio]{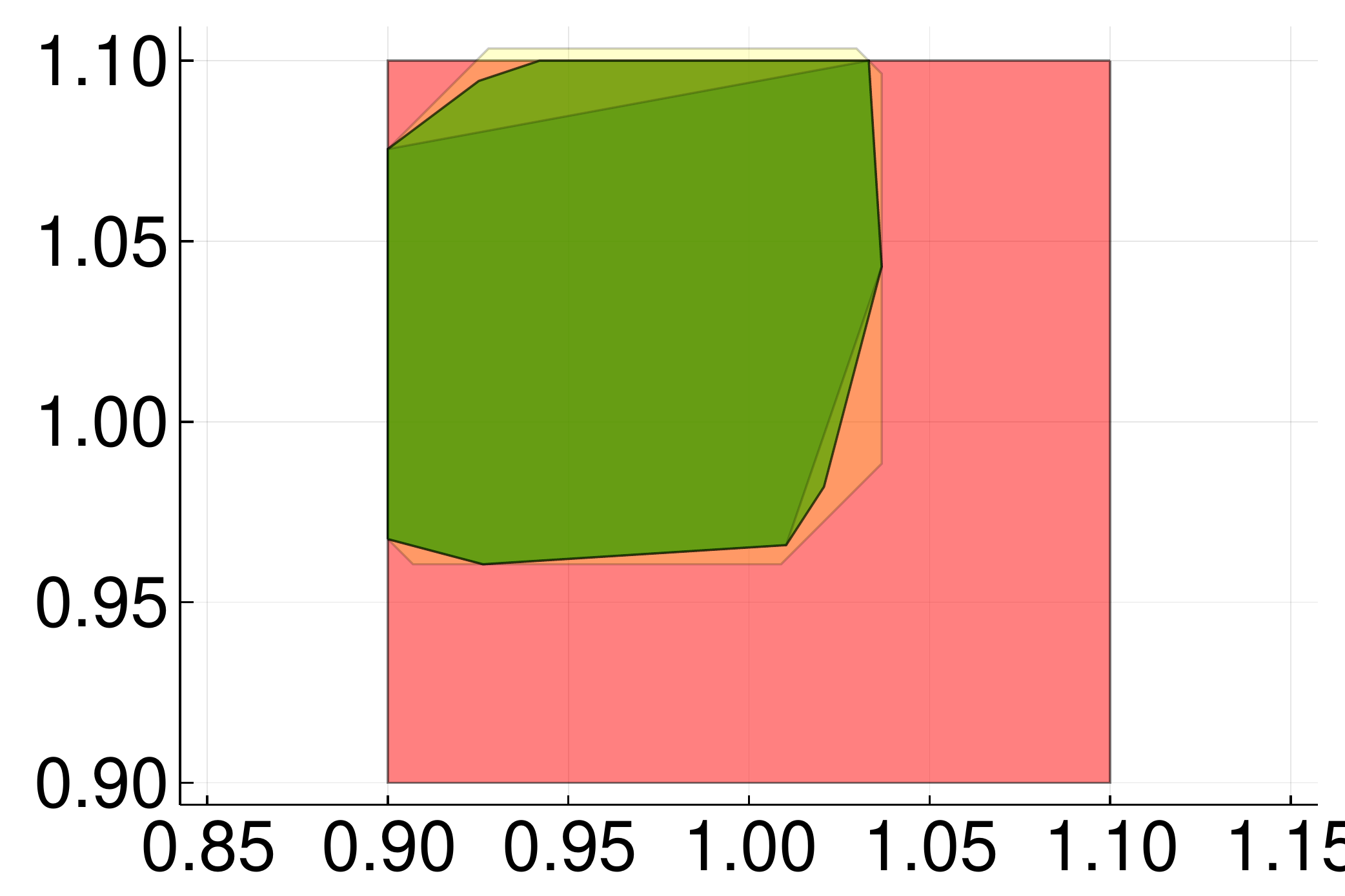}
		\caption{Analysis for $\alpha = 0.01$ and a single pieces of duration $4\pi$.}
		\label{fig:membership_example_different_1}
	\end{subfigure}
	
	\begin{subfigure}[b]{\linewidth}
		\centering
		\includegraphics[width=\linewidth,height=6cm,keepaspectratio,clip,trim=0 115mm 0 0]{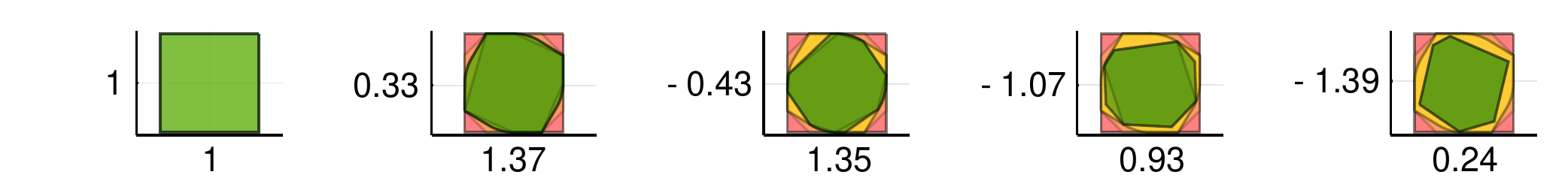}
		\includegraphics[width=\linewidth,height=6cm,keepaspectratio,clip,trim=0 0 0 115mm]{membership_same_23}
		\caption{Analysis for $\alpha = 0$ and $23$ pieces of uniform 
		duration (only the first and last four intermediate results are shown).
		The last subplot shows $1{,}000$ samplings from the final dark 
		yellow region.}
		\label{fig:membership_example_same_23}
	\end{subfigure}
	\caption{Analysis results for systems~\eqref{eq:system1} and~\eqref{eq:system2} with $\eps = 0.1$ and $\alpha = 0.01$ or $\alpha = 0$, respectively.}
	\label{fig:membership_example}
\end{figure}

\section{Model update}\label{sec:update}

In this section we describe a procedure to solve
Problem~\ref{problem:synthesis} and propose a minimality criterion.
The procedure tackles the problem by evaluating the given \pwa trajectories in 
an online fashion.

\subsection{Lexicographic order to rank model updates}
For a given \adha $\H$, a \pwa trajectory $f$, and a value $\eps \in \Rnn$, our 
procedure searches for a path $\pi$ in $\H$ such that the membership query of 
$f$ in $\H$ is positive. If there is no such path in $\H$,
the procedure modifies $\H$ such that the modified \adha includes such a path.
The path selection and the corresponding modifications of the \adha
are chosen in the following order:
\begin{enumin}{,}{and}
	\item increasing invariants and guards
	\item adding new transitions
	\item adding new locations
\end{enumin}
The rationale is to keep the number of locations as small as possible.

Formally, we define a tuple $\mathit{mod} = (n_l, n_t, n_c)$ that
keeps track of the above modifications, where
$n_l$ tracks the number of new locations,
$n_t$ tracks the number of new transitions, and
$n_c$ tracks the number of modified constraints (invariants and guards).
We will use the tuple $\mathit{mod}$ for path selection in a lexicographic
order where we aim to find the minimal tuple.
For instance, the tuple $(0,0,0)$, representing no modifications at all,
is selected over any other tuple; the tuple $(0,2,0)$, representing two transition
additions, is selected over the tuple $(1,0,0)$, representing a location addition.

\subsection{Online model update}

For a given \adha $\H$, a \pwa trajectory $f$, and a value $\eps \in \Rnn$, we 
update $\H$ for each affine piece in $f$ if required. We describe
how $\H$ is modified for a concrete piece of $f$ according to a location $q$
that may either be part of $\H$ or be a new location to be added to $\H$.

\begin{definition}
	Consider an \adha $\H$, a path $\pi$ in $\H$ with $\len{\pi} = k-1$, an existing 
	or new location $q$, a \pwa trajectory $f$, represented by the tuple $(\{ [0, 
	t_1],\ldots, [t_{k-1}, t_k] \},$ $ \{ A_1, \ldots, A_k \},$ $ \{ \b_1, \ldots, \b_k 
	\})$, a polyhedron $P$, and a value $\eps \in \Rnn$.
	A \emph{$q$-update of $\H$ with respect to 
	$f$, $P$, and $\eps$}  is an \adha $\H'$, denoted by 
	$\modha_{q,f}^{P,\eps}(\H)$, 
	such that $\Q' = \Q \cup \{ q \}$, $\E' = \E \cup \{ (\last{\pi}, q) \}$,
	and the remaining components depend on whether $q$ exists in $Q$ or is
	a new location.
	\\[1mm]
	If $q \in Q$:
	\begin{enumin}{,}{and}
		\item
		$\Flow' \equiv \Flow$
		\item
		$\rest{\Inv'}{\Q \setminus \{q\}} \equiv \Inv$ and
		$\Inv'(q) =$ $\chull(\Inv(q) \cup \rset_I)$
		\item
		$\rest{\G'}{\E \setminus \{ (\last{\pi}, q) \}} \equiv \G$ and \linebreak
		$\G'(\last{\pi}, q) = $ $\chull(\G(\last{\pi}, q) \cup \rset_G)$
	\end{enumin}
	\\[1mm]
	If $q \notin \Q$:
	\begin{enumin}{,}{and}
		\item
		$\rest{\Flow'}{\Q} \equiv \Flow$ and $\Flow'(q) = (A_k, \b_k)$
		\item
		$\rest{\Inv'}{\Q} \equiv \Inv$ and $\Inv'(q) = \chull(\rset_I)$
		\item
		$\rest{\G'}{\E} \equiv \G$ and
		\linebreak
		$\G'(\last{\pi}, q) = \rset_G$
	\end{enumin}
	Here $\rset_I = \bigcup_{t \in [t_{k-1}, t_k]} \reach(P_{k-1},\pi \cdot q,
	\linebreak[1]
	\rest{f}{[t_{k-1}, t]}, \eps)$, with $P_{k-1}$ as defined in \eqref{eq:syncset}, and $\rset_G = \tube(f,\eps)(t_k)$.
\end{definition}
Observe that $\exec(\H) \subseteq \exec(\modha_{q,f}^{P,\eps}(\H))$.
Next we define a tree capturing the \adha updates for every affine 
piece in $f$.

\begin{definition}
	Given an $n$-dimensional \adha $\H_0$ and a \pwa trajectory $f$ with $k$ 
	pieces, an \emph{exploration tree for $\H_0$ and $f$} is $\T 
	= (\V, \E)$ with $k$ layers (not counting the root node as a layer).
	Each node $\nu \in \V$ is represented as a tuple 
	$(\pi, \H, \mathit{mod}, \s)$ where
	$\pi$ is a path in an \adha $\H$,
	$\mathit{mod}$ is a triple of integers $(n_l, n_t, n_c)$, and
	$\s$ is a four-valued variable called \emph{status}
	(with meanings
	$0$\emph{:}~`unexplored',
	$1$\emph{:}~`activated',
	$2$\emph{:}~`explored',
	$3$\emph{:}~`deactivated').
\end{definition}
Observe that exploration trees for $\H_0$ and $f$ can only differ in the status.
The set of all exploration trees for $\H_0$ and $f$
is denoted by $\treeset{\H_0,f}$, and we call all trees
belonging to $\treeset{\H_0,f}$ \emph{similar}.
We may add a subscript to the elements in the node $\nu = (\pi,\H,\mathit{mod}, 
\s)$
(i.e., write $\pi_\nu$ etc.) for clarity.
We define, for an initial 
polyhedron $P$ and a value $\eps \in \Rnn$, an exploration tree $\T_0^{P,\eps} 
\in 
\treeset{\H_0,f}$ such that the 
\emph{root node} is
$([\,], \H_0, (0,0,0), 0)$, where $[\,]$ is the empty path.
Each node $(\pi, \H, (n_l, n_t, n_c), 0)$ in layer $i-1$, for $0 < i \leq k$, where $\Q_\H 
= \{ q_1, \ldots, q_m \}$, has $m+1$ \emph{child nodes}. The first $m$ nodes are:
\begin{align*}
	(\pi \cdot q_1, \modha_{q_1,f}^{P,\eps}(\H), (n_l, n_t + a_1, n_c + b_1), 0), 
	\ldots,\\
	(\pi \cdot q_m, \modha_{q_m,f}^{P,\eps}(\H), (n_l, n_t + a_m, n_c + b_m), 0),
\end{align*}
where $a_j = 0$ if $(\last{\pi}, q_j) \in \E_\H$
and $a_j  = 1$ otherwise, and $b_j$ is the number of constraint modifications
for invariants and guards with respect to $\H$, for every $1 \leq j \leq m$.
The last child node is:
\begin{align*}
	(\pi \cdot q, \modha_{q,f}^{P,\eps} (\H), (n_l +1, n_t + 1, n_c), 0),
\end{align*}
where $q$ is a new location with $\Flow(q) = (A_i, \b_i)$.

The paths from root to leaves in an exploration tree represent
all possible paths in updated \adha{s}, given the initial \adha $\H_0$, for
exploring membership of $f$.
An upper bound on the number of paths is $(m+k)^k$, where $m = |Q_0|$ is the number of locations in $\H_0$ and $k$ is the number of pieces in $f$.
The complexity for the membership query is in $\mathcal{O}(p(n))$ for some polynomial $p$ in the dimension $n$.
Hence the complexity for a membership check in each path of
the exploration tree is upper-bounded by $\mathcal{O}((m+k)^k k p(n))$.

We introduce a strategy for partial exploration that minimizes automaton
modifications (according to $\mathit{mod}$).
Given $\H_0$ and $f$, a \emph{decision strategy} is a function 
$D:  \treeset{\H_0,f} \times \V \to \V$ that determines the next node 
to be analyzed in an exploration tree.
A decision strategy is combined with a tree update in order to activate and 
explore nodes or discard useless nodes.
We say that a node is \emph{unexplored} when its status is $0$. 
We can explore a node when it is \emph{activated} (status $1$).
After exploration, if the membership query is positive, we set 
the status to $2$ (\emph{explored}) and otherwise to $3$ (\emph{deactivated}). 
Child nodes of deactivated nodes need not be explored further.

\begin{definition}
	Given an \adha $\H_0$, a \pwa trajectory $f$, a polyhedron $P$, and a value 
	$\eps \in \Rnn$, 
	an \emph{\eps-tree update} function, $\upd_{\eps}: \treeset{\H_0,f} \times \V 
	\to 
	\treeset{\H_0,f}$, 
	maps a tree $\T$ and a node $\nu$ in the $i$-th layer to a similar tree 
	such that
	$\s_\nu 
	= 3 $ if $\reach(P,\pi_\nu, \rest{f}{[0,t_i]},\eps) = \emptyset$, 
	and
	$\s_\nu = 2$ and $\s_{\nu'} = 1$  for every  node $\nu' \in \children{\nu}$
	otherwise, and leaves the status of all other nodes unchanged. 
\end{definition}

Given a set of nodes $W$, we  
denote by $\activated{W}$ the set of nodes with activated status, i.e., $\{ \nu 
\in 
W : \s_\nu = 1 \}$. Our decision strategy minimizing 
$\mathit{mod}$ is
$D(\T, \nu) = \displaystyle\arg \min_{\nu' \in \, \activated{\V}} 
\mathit{mod}_{\nu'},$
assuming that $\arg \min$ returns one node if several nodes minimize
the $\mathit{mod}$ value.

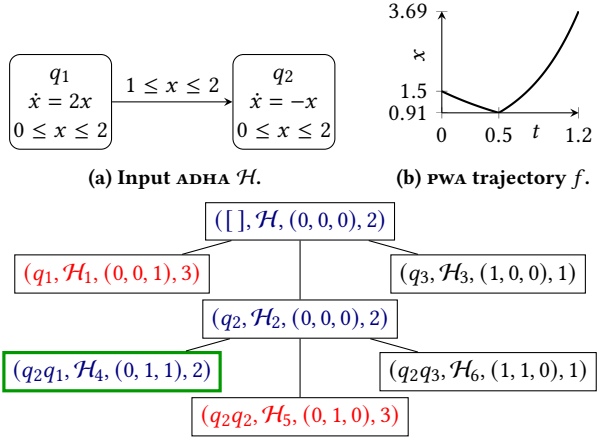
\begin{figure}
	\centering
	\begin{subfigure}[b]{.59\linewidth}
		\centering
		\begin{tikzpicture}
	\node[loc] (q1) {\begin{tabular}{@{} c @{}}$q_1$ \\ $\dot{x} = 2x$ \\ $0 \leq x \leq 2$\end{tabular}};
	\node[loc,right=16mm of q1] (q2) {\begin{tabular}{@{} c @{}}$q_2$ \\ $\dot{x} = -x$ \\ $0 \leq x \leq 2$\end{tabular}};
	\draw[t] (q1) to node[above] {$1 \leq x \leq 2$} (q2);
\end{tikzpicture}
		\caption{Input \adha $\H$.}
		\label{fig:automaton_example}
	\end{subfigure}
	\begin{subfigure}[b]{.39\linewidth}
		\centering
		\raisebox{-4mm}{\begin{tikzpicture}
	\begin{axis}[width=34mm,xtick={0, 0.5, 1.2},ytick={0.91, 1.5, 3.69},xtick align=center,ytick align=center,axis x line=bottom,axis y line=left,xlabel=\hspace*{7mm}$t$,ylabel=\hspace*{2mm}$x$,xlabel shift=-4.5mm,ylabel shift=-7mm]
		\addplot[thick,name path=first][domain=0:0.5] {1.5*e^(-x)};
		\addplot[thick,name path=second][domain=0.5:1.2] {1.5*e^(-0.5)*e^(2*(x-0.5))};
	\end{axis}
\end{tikzpicture}}
		\caption{\pwa trajectory $f$.}
		\label{fig:pwld_example}
	\end{subfigure}
	
	\vspace*{1mm}
	
	\begin{subfigure}{\linewidth}
		\centering
		\begin{tikzpicture}[sibling distance=25mm]
\node[tn] (n01) {\color{blue!50!black} $([\,], \H, (0, 0, 0), 2)$}
child { node[tn,yshift=8mm] {\color{red} $(q_1, \H_1, (0, 0, 1), 3)$} }
child { node[tn,yshift=2mm] {\color{blue!50!black} $(q_2, \H_2, (0, 0, 0), 2)$}
	child { node[tn,very thick,color=green!60!black,yshift=8mm] 
	{\color{blue!50!black} $(q_2 q_1, \H_4, (0, 1, 1), 2)$} }
	child { node[tn,yshift=2mm] {\color{red} $(q_2 q_2, \H_5,  (0, 1, 0), 3)$} }
	child { node[tn,yshift=8mm] {$(q_2 q_3, \H_6, (1, 1, 0), 1)$} } }
child { node[tn,yshift=8mm] {$(q_3, \H_3, (1, 0, 0), 1)$} };
\end{tikzpicture}
		\caption{Partial view of a search tree for $\eps = 0.1$.
			Red nodes are deactivated, blue nodes have been explored, and black 
			nodes are activated (and we omit their child nodes).
			The algorithm selects the path to the green node for a membership query.}
		\label{fig:search_tree}
	\end{subfigure}
	\caption{An \adha $\H$, a \pwa trajectory $f$ with two pieces, and 
	a (partial) exploration tree for $\H$ and $f$.}
	\label{fig:search_tree_example}
\end{figure}

\begin{example}
	Figure~\ref{fig:search_tree_example} shows an example of an intermediate state of an exploration tree for a given \adha and a \pwa trajectory with two pieces.
	The root node has been described before.
	For the remaining tree nodes we represent the
	automata $\H_i$ only symbolically.
	The first piece of $f$ follows the dynamics $\dot{x} = -x$ for $0.5$ time units.
	The available choices for the first automaton mode are $q_1$, $q_2$, or a new mode $q_3$; hence the root node expands to three new nodes.
	The node with path $q_1$ requires a modification of the invariant of $q_1$ because dwelling in that mode for $0.5$ time units is not possible otherwise.
	However, since the final reachable states do not intersect with the \eps-tube 
	(negative membership query), 
	this node status is set to $3$ (deactivated) and none of the child
	nodes are explored further.
	The node with the new location $q_3$ has a ``location entry'' in the 
	modification tuple.
	The node with path $q_2$ does not require any modifications (i.e., $\H_2 = 
	\H$) and is hence chosen as the next node for exploration.
	Now we consider the second piece with dynamics $\dot{x} = 2x$ for $0.7$ time units.
	Again we have the choice between the existing locations and a new location 
	$q_3$.
	The exploration works like before, only that this time we need to add a 
	transition from $q_2$ to the next location in all three cases (since $q_2$ does 
	not have any outgoing transitions in $\H_2 = \H$ yet).
	The $2$-leaf with the path $q_2q_1$ has the highest priority and we perform a membership query for it.
	In this case, the query returns a positive answer and the algorithm outputs the 
	automaton $\H_4$, which looks like $\H$ but with an additional transition and 
	an extended invariant in location $q_1$.
\end{example}

\begin{algorithm}[t]
	\caption{Hybrid model update}
	\label{alg:update}
	\begin{algorithmic}[1]
		\REQUIRE An \adha $\H$, a \pwa trajectory $f$ and $\eps \in \Rnn$.
		
		\ENSURE An \adha $\H'$ such that it $\eps$-captures $f$.
		
		\STATE $\T$ := $\initialize(\H,f,\eps)$
		
		\STATE $\nu$ := $\rootn(\T)$ \label{line:rootnode}
		
		\STATE $\T$ := $\upd_\eps(\T,\nu)$
		
		\WHILE{$\nu$ not in bottom layer with $\s_\nu = 2$}
			
			\STATE $\nu$ := $D(\T, \nu)$ \label{line:nu}
			
			\STATE $\T$ := $\upd_\eps(\T,\nu)$ \label{line:T}
			
			\STATE $\H'$ := $\H_\nu$
			
		\ENDWHILE
		
		\RETURN{$\H'$}
	\end{algorithmic}
\end{algorithm}
\begin{algorithm}[t]
	\caption{Synthesis of \adha from \pwa trajectories}
	\label{alg:synthesis}
	\begin{algorithmic}[1]
		\REQUIRE A finite set of \pwa trajectories $F$ and $\eps \in \Rnn$.
		
		\STATE $\H$ := $\emptyset$ \COMMENT{empty automaton with no 
		location}
		
		\FOR{$f$ in $F$}
			
			\STATE $\H$ :=  Algorithm~\ref{alg:update}($\H, f, \eps$)
			
		\ENDFOR
		
		\RETURN{$\H$}
	\end{algorithmic}
\end{algorithm}

Algorithm~\ref{alg:update} shows the procedure for a model update given an 
initial \adha $\H$, a \pwa trajectory $f$, and a value $\eps \in \Rnn$. 
The function $\initialize$ constructs the exploration tree $\T_0^{P,\eps}$ for the 
polyhedron $P = \tube(f,\eps)(0)$.
Then the algorithm starts 
exploring from the root node (line~\ref{line:rootnode}) and subsequently explores nodes driven 
by the decision strategy (line~\ref{line:nu}), which chooses activated nodes with minimum 
$\mathit{mod}$ component and iteratively activates every child nodes and deactivates the current node or sets it to explored (line~\ref{line:T}).
The algorithm returns the updated 
model $\H'$.
Finally, Problem~\ref{problem:synthesis} is solved by iteratively running 
Algorithm~\ref{alg:update} over every trajectory $f \in F$ and
modifying the \adha, as shown in Algorithm~\ref{alg:synthesis}.

\begin{proposition}
	Given an \adha $\H$, a \pwa trajectory $f$, and a value $\eps \in \Rnn$, 
	Algorithm~\ref{alg:update} provides an updated \adha $\eps$-capturing $f$ 
	and minimizing the number of modifications.
\end{proposition}
\begin{proof}
	Given an \adha $\H$, a \pwa trajectory $f$, and a value $\eps \in 
	\Rnn$, Algorithm~\ref{alg:update} proceeds as follows.
	First, the algorithm constructs an initial exploration tree 
	$\T_0^{P,\eps}$ (line $1$) whose nodes contain all the possible 
	modifications of $\H$ with unexplored status ($0$).
	Then, the algorithm sets $\nu$ as the root node (line $2$) and 
	applies the \eps-tree update over the initial exploration tree and 
	$\nu$ (line $3$). This update sets the status for all nodes in the 
	first layer to $1$ and for the root node to $2$ because 
	$\reach(P,\emptyset, \rest{f}{[0,0]},\eps) = P$.
	Next, the algorithm iterates (line $4$) as follows.
	The decision strategy $D(\T, \nu)$ selects the node of the search 
	tree with minimum $\mathit{mod}$ value and activated status 
	($1$).
	Then, the \eps-tree update requires to check if the \adha 
	$\H_\nu$ in the $k$-th layer of the exploration tree 
	\eps-captures the $k$ first pieces of $f$.
	If these pieces are not captured, the status of the 
	node is deactivated (set to $3$) and all child nodes will remain 
	unexplored (with status $0$) forever.
	If $\H_\nu$ \eps-captures $\rest{f}{[0,t_k]}$, the 
	status of the node is set to $2$ (explored) and the status of all child nodes
	is activated ($1$).
	The loop runs until $\nu$ is a node at the 
	bottom layer with status $2$, which means that $\H_\nu$ 
	\eps-captures $f$ and that $\mathit{mod}$ is minimum due to
	the decision strategy.
	The algorithm terminates because, in the 
	worst case, it will choose the path in the search tree 
	where a new mode is added for each piece of $f$;
	clearly the
	\adha at the leaf of that path \eps-captures $f$.
\end{proof}

\begin{theorem}
	Given an \adha $\H$, a set $F$ of \pwa trajectories, and a value $\eps \in 
	\Rnn$, Algorithm~\ref{alg:synthesis} solves Problem~\ref{problem:synthesis}.
\end{theorem}

\section{Implementation and case study}\label{sec:evaluation}

In this section we describe our implementation and evaluate it:
in the first two examples we obtain the \pwa trajectories from
random executions with perturbed dynamics from a given \adha model;
in a third example we construct the \pwa trajectories from time series.

\subsection{Implementation}\label{sec:implementation}

We implemented our approach in \emph{HySynth}~\cite{hysynth} where
we wrote the high-level synthesis algorithm in Python and the low-level
algorithms in Julia.
For the ODE optimization (both in Section~\ref{sec:ts2function} and Algorithm~\ref{alg:linear_singleton}) we
use the libraries \texttt{Optim.jl}~\cite{mogensen2018optim}
(which uses Brent's method~\cite{Brent71} to find a root in
a bracketing interval and guarantees convergence for functions computable
within the interval) and
\texttt{DifferentialEquations.jl}~\cite{rackauckas2017differentialequations}
as follows (assuming linear dynamics
without loss of generality).
Given two $n$-dimensional executions $\dot{\x} = A\x$, $\x(0) = \x_0$ and
$\dot{\y} = B\y$, $\y(0) = \y_0$, we construct a $3n$-dimensional execution
$\dot{\z} = C\z$, $\z(0) = \z_0$ where
\[
	C = \begin{pmatrix} A & 0 & 0 \\ 0 & B & 0 \\ A & -B & 0 \end{pmatrix}
	\qquad
	\z_0 = \begin{pmatrix} \x_0 \\ \y_0 \\ \x_0 - \y_0 \end{pmatrix}.
\]
We are interested in the projection of $\z(t)$ onto the last $n$ dimensions.
Calling this projection $\w(t)$, the norm of $\w(t)$ describes the distance
between $\x(t)$ and $\y(t)$, i.e., $\norm{\w(t)} = d(\x(t), \y(t))$.
We query the solver for each dimension of $\w(t)$ to find the maximum distance.

We use \texttt{JuliaReach}~\cite{BogomolovFFPS19} for the set computations and
reachability analysis.
As mentioned in Section~\ref{sec:refinement}, the polytopes over- and
under-approximating the synchronized reachable sets constructed during the
membership query grow in complexity, especially for input trajectories with many
pieces.
We simplify the sets after each piece, i.e., we under-approximate an
under-approximation (for which \texttt{JuliaReach} computes a polytope from
support vectors in template directions) and over-approximate an
over-approximation (for which we implemented an algorithm from~\cite{GuibasNZ03}
to compute a zonotope in template directions) with octagonal directions
(i.e., axis-parallel or diagonal constraints in two dimensions).

\subsection{Evaluation}

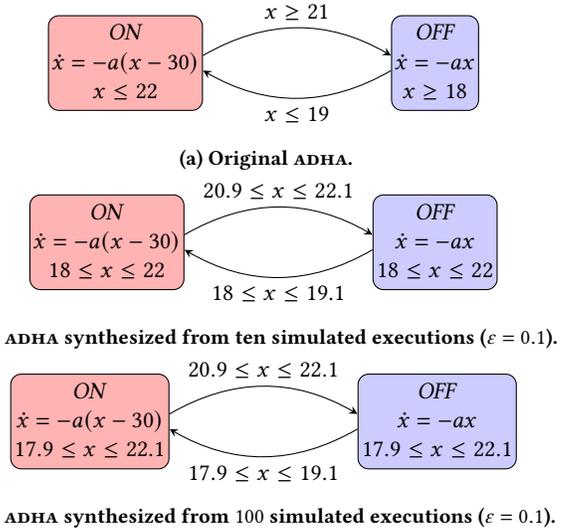
\begin{figure}
	\centering
	\begin{subfigure}{\linewidth}
		\centering
		\begin{tikzpicture}
	\node[loc,fill=\Con] (on) {\begin{tabular}{@{} c @{}}\textit{ON} \\ $\dot{x} = -a(x-30)$ \\ $x \leq 22$\end{tabular}};
	\node[loc,fill=\Coff,right=25mm of on] (off) {\begin{tabular}{@{} c @{}}\textit{OFF} \\ $\dot{x} = -ax$ \\ $x \geq 18$\end{tabular}};
	\draw[t,bend left] ($(on.east) + (0,1mm)$) to node[above] {$x \geq 21$} ($(off.west) + (0,1mm)$);
	\draw[t,bend left] ($(off.west) + (0,-1mm)$) to node[below] {$x \leq 19$} ($(on.east) + (0,-1mm)$);
\end{tikzpicture}
		\caption{Original \adha.}
		\label{fig:thermostat_original}
	\end{subfigure}
	\begin{subfigure}{\linewidth}
		\centering
		\begin{tikzpicture}
	\node[loc,fill=\Con] (on) {\begin{tabular}{@{} c @{}}\textit{ON} \\ $\dot{x} = -a(x-30)$ \\ $18 \leq x \leq 22$\end{tabular}};
	\node[loc,fill=\Coff,right=25mm of on] (off) {\begin{tabular}{@{} c @{}}\textit{OFF} \\ $\dot{x} = -ax$ \\ $18 \leq x \leq 22$\end{tabular}};
	\draw[t,bend left] ($(on.east) + (0,1mm)$) to node[above] {$20.9 \leq x \leq 22.1$} ($(off.west) + (0,1mm)$);
	\draw[t,bend left] ($(off.west) + (0,-1mm)$) to node[below] {$18 \leq x \leq 19.1$} ($(on.east) + (0,-1mm)$);
\end{tikzpicture}
		\caption{\adha synthesized from ten simulated executions ($\eps = 0.1$).}
		\label{fig:thermostat_synthesized_10}
	\end{subfigure}
	\begin{subfigure}{\linewidth}
		\centering
		\begin{tikzpicture}
	\node[loc,fill=\Con] (on) {\begin{tabular}{@{} c @{}}\textit{ON} \\ $\dot{x} = -a(x-30)$ \\ $17.9 \leq x \leq 22.1$\end{tabular}};
	\node[loc,fill=\Coff,right=25mm of on] (off) {\begin{tabular}{@{} c @{}}\textit{OFF} \\ $\dot{x} = -ax$ \\ $17.9 \leq x \leq 22.1$\end{tabular}};
	\draw[t,bend left] ($(on.east) + (0,1mm)$) to node[above] {$20.9 \leq x \leq 22.1$} ($(off.west) + (0,1mm)$);
	\draw[t,bend left] ($(off.west) + (0,-1mm)$) to node[below] {$17.9 \leq x \leq 19.1$} ($(on.east) + (0,-1mm)$);
\end{tikzpicture}
		\caption{\adha synthesized from $100$ simulated executions ($\eps = 0.1$).}
		\label{fig:thermostat_synthesized_100}
	\end{subfigure}
	\caption{\adha models of the heater system. Numbers are rounded to one decimal place.}
\end{figure}

We consider an \adha that models a \textbf{heater} with two locations
``\textit{ON}'' and ``\textit{OFF}'', as depicted in Figure~\ref{fig:thermostat_original} with
parameter value $a = 0.1$.

Next we describe how we sampled executions from the model.
The inputs to the simulation procedure are
\begin{enumin}{,}{,and}
	\item an \adha (here: the heater model)
	\item a desired path length (here: $6$)
	\item a maximum dwell time per location (here: $7$)
	\item a time step (here: $0.05$)
	\item a maximum perturbation (here: $0.001$)	
\end{enumin}
We first sample an initial location $q_0$ and an initial (continuous) state
$\x_0$ from $\Inv(q_0)$.
Then we repeat the following loop.
Given a location $q$ and a state $\x$, we first compute a matrix $A$ by
perturbing the dynamics matrix $\Flow(q)$ (technically, we only perturb
non-zero entries).
Then we compute the discrete-time successor of $\x$ with the fixed time step
$t$ (via $\x' := \lreach(\{\x\}, A, t)$), and we check which of the outgoing
transitions of $q$ are enabled for this new state.
We continue computing successor states and collecting enabled transitions
until either the state leaves $\Inv(q)$ or we exceed the maximum dwell time.
Then we choose a random transition together with a random time point of those
that were enabled.
The above loop terminates if either there is no transition enabled or we
exceed the desired path length.

We applied the above procedure to obtain 100 random executions from the
heater model.
Then we first learned a model from the first ten executions and then
continued modifying the resulting \adha with the remaining 90
executions, where we used a precision value $\eps = 0.1$.
(Note that our algorithmic framework behaves exactly the same way as if
we had learned an \adha from the 100 executions at once.
The split into two stages is only for illustrative purposes.)
We show the intermediate and the final result obtained with our implementation
in Figure~\ref{fig:thermostat_synthesized_10} and Figure~\ref{fig:thermostat_synthesized_100} respectively, and random simulations
in Figure~\ref{fig:thermostat_simulations}.

\begin{figure}
	\begin{subfigure}[t]{.48\linewidth}
		\centering
		\includegraphics[width=\textwidth,keepaspectratio]{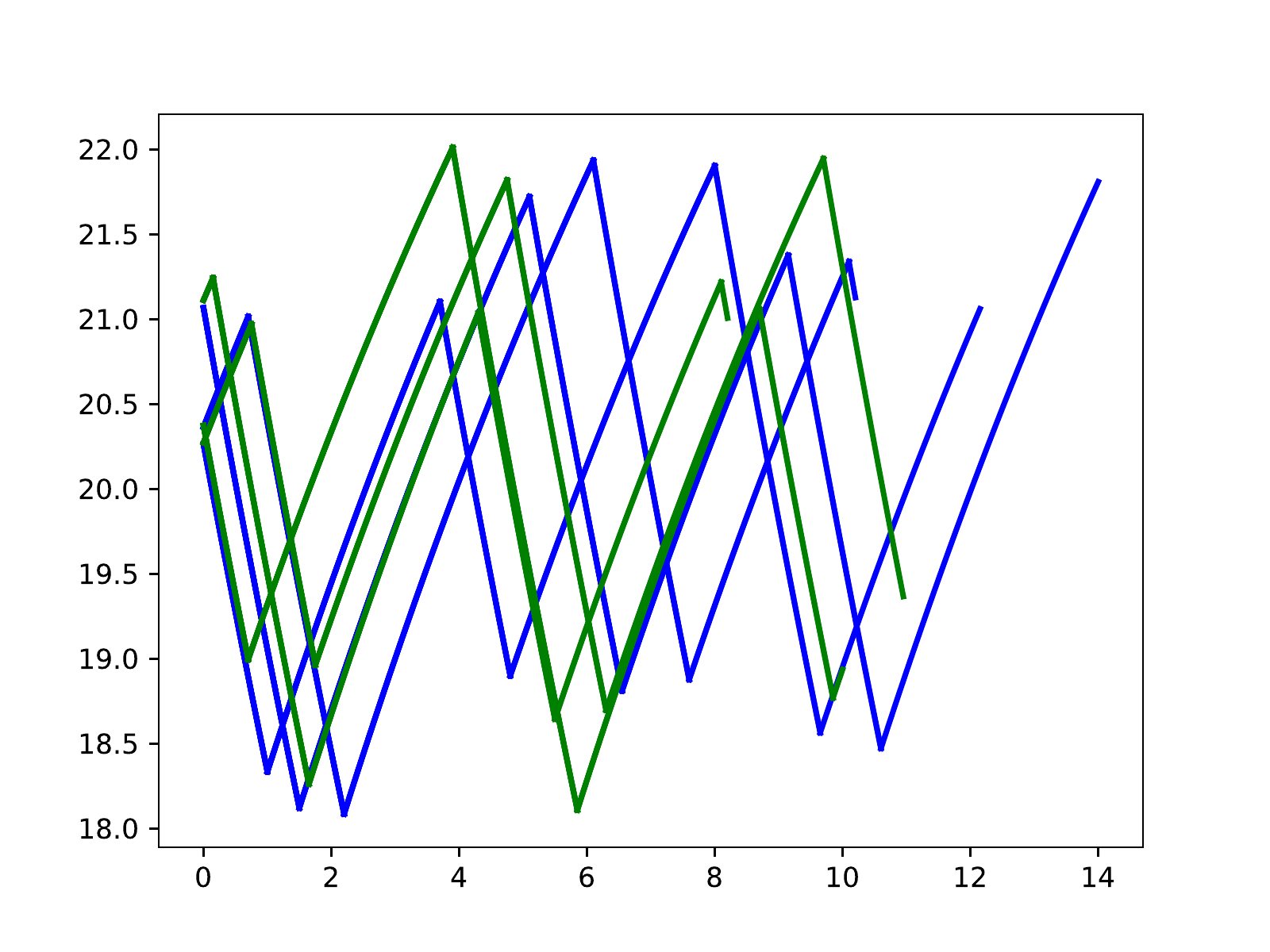}
		\caption{Heater model: Variable $\x$ (ordinate) over time (abscissa).}
		\label{fig:thermostat_simulations}
	\end{subfigure}
	\hfill
	\begin{subfigure}[t]{.48\linewidth}
		\centering
		\includegraphics[width=\textwidth,keepaspectratio]{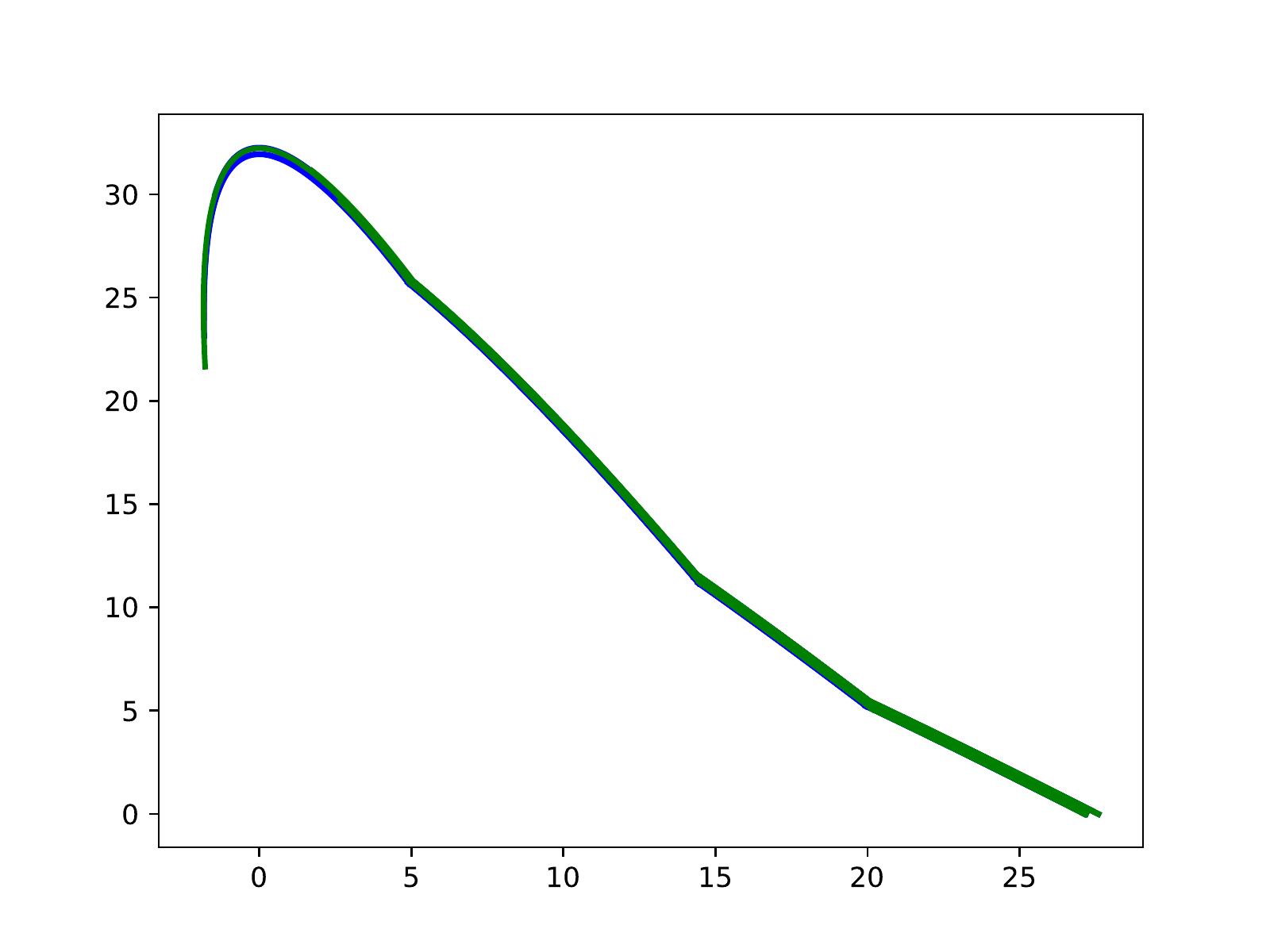}
		\caption{Gearbox model: Phase portrait.}
		\label{fig:gearbox_simulations}
	\end{subfigure}
	\caption{Random simulations:
	Three blue simulations are obtained from the original model and three green simulations are obtained from the synthesized model.}
	\label{fig:simulations}
\end{figure}

The first observation is that the discrete structure of the resulting
\adha matches exactly the structure of the original model.
The reason why the dynamics matrices of the locations is the same as in
the original model is because we did not perturb the dynamics of the
very first execution in order to obtain a legible flow representation.
Still, even though the algorithm is confronted with slightly different
dynamics in all other executions, it does not add further locations to the
\adha, thanks to the precision value \eps.
As can be seen, the invariants and guards in the final \adha
over-approximate the original guards by \eps, which is expected by construction.

\begin{figure}
	\centering
	\begin{subfigure}{\linewidth}
		\centering
		\begin{tikzpicture}
	\node[loc,fill=\cI] (q1) {\begin{tabular}{@{} c @{}}$\dot{\x} = A_1\x$ \\ $v \geq 20$\end{tabular}};
	\node[loc,fill=\cII,right=22mm of q1] (q2) {\begin{tabular}{@{} c @{}}$\dot{\x} = A_2\x$ \\ $14 \leq v \leq 23$\end{tabular}};
	\node[loc,fill=\cIII,right=4mm of q2,yshift=-6mm] (q3) {\begin{tabular}{@{} c @{}}$\dot{\x} = A_3\x$ \\ $5 \leq v \leq 19$\end{tabular}};
	\node[loc,fill=\cIV,left=27mm of q3] (q4) {\begin{tabular}{@{} c @{}}$\dot{\x} = A_4\x$ \\ $v \leq 13$\end{tabular}};
	\draw[t] (q1) to node[above] {$v = 20$} (q2);
	\draw[t,bend left] (q2) to node[above right=-1mm] {$v = 14$} (q3);
	\draw[t] (q3) to node[below] {$v = 5$} (q4);
\end{tikzpicture}
		\caption{Original \adha.}
		\label{fig:gearbox_original}
	\end{subfigure}
	
	\begin{subfigure}{\linewidth}
		\centering
		\begin{tikzpicture}
	\node[loc,fill=\cI] (q1) {\begin{tabular}{@{} c @{}}$\dot{\x} = A_1\x$ \\ $20 \leq v \leq 28$ \\ $0 \leq w \leq 6$\end{tabular}};
	\node[loc,fill=\cII,right=24mm of q1] (q2) {\begin{tabular}{@{} c @{}}$\dot{\x} = A_2\x$ \\ $14 \leq v \leq 20$ \\ $4 \leq w \leq 12$\end{tabular}};
	\node[loc,fill=\cIII,right=4mm of q2,yshift=-8mm] (q3) {\begin{tabular}{@{} c @{}}$\dot{\x} = A_3\x$ \\ $5 \leq v \leq 15$ \\ $10 \leq w \leq 26$\end{tabular}};
	\node[loc,fill=\cIV,left=25mm of q3] (q4) {\begin{tabular}{@{} c @{}}$\dot{\x} = A_4\x$ \\ $-2 \leq v \leq 5$ \\ $18 \leq w \leq 32$\end{tabular}};
	\draw[t] (q1) to node[above=-1mm] {\begin{tabular}{@{} c @{}}$20 \leq v \leq 20$ \\ $4 \leq w \leq 6$\end{tabular}} (q2);
	\draw[t,bend left] (q2) to node[above right=-2mm] {\begin{tabular}{@{} c @{}}$14 \leq v \leq 15$ \\ $10 \leq w \leq 12$\end{tabular}} (q3);
	\draw[t] (q3) to node[below=-1mm] {\begin{tabular}{@{} c @{}}$5 \leq v \leq 5$ \\ $25 \leq w \leq 26$\end{tabular}} (q4);
\end{tikzpicture}
		\caption{\adha synthesized from ten simulated executions ($\eps = 0.1$).}
		\label{fig:gearbox_synthesized_10}
	\end{subfigure}
	\caption{\adha models of the gearbox system.
	Numbers are rounded to integers; constraints are approximated by boxes.}
	\label{fig:gearbox}
\end{figure}
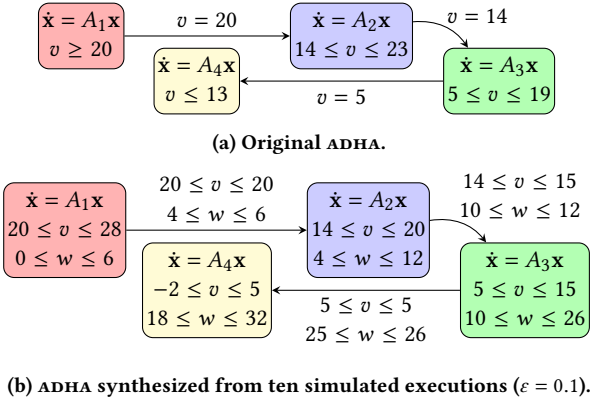

We also applied the algorithm to a two-dimensional \textbf{gearbox} model with 
variables
$v$ and $w$ from~\cite{PrabhakarS16} and we refer to that reference for further
details about the model.
We present the results for 10 simulations, a maximum perturbation of $0.0001$,
and initial states sampled from the red location and the set
$26 \leq v(0) \leq 28, w(0) = 0$ in Figure~\ref{fig:gearbox},
and random simulations in Figure~\ref{fig:gearbox_simulations}.
Overall we see a similar algorithmic performance as for the heater model.

\begin{figure}
	\centering
	\begin{subfigure}{.48\linewidth}
		\centering
		\includegraphics[width=\textwidth,keepaspectratio,clip,trim=12mm 0mm 16mm 0mm]{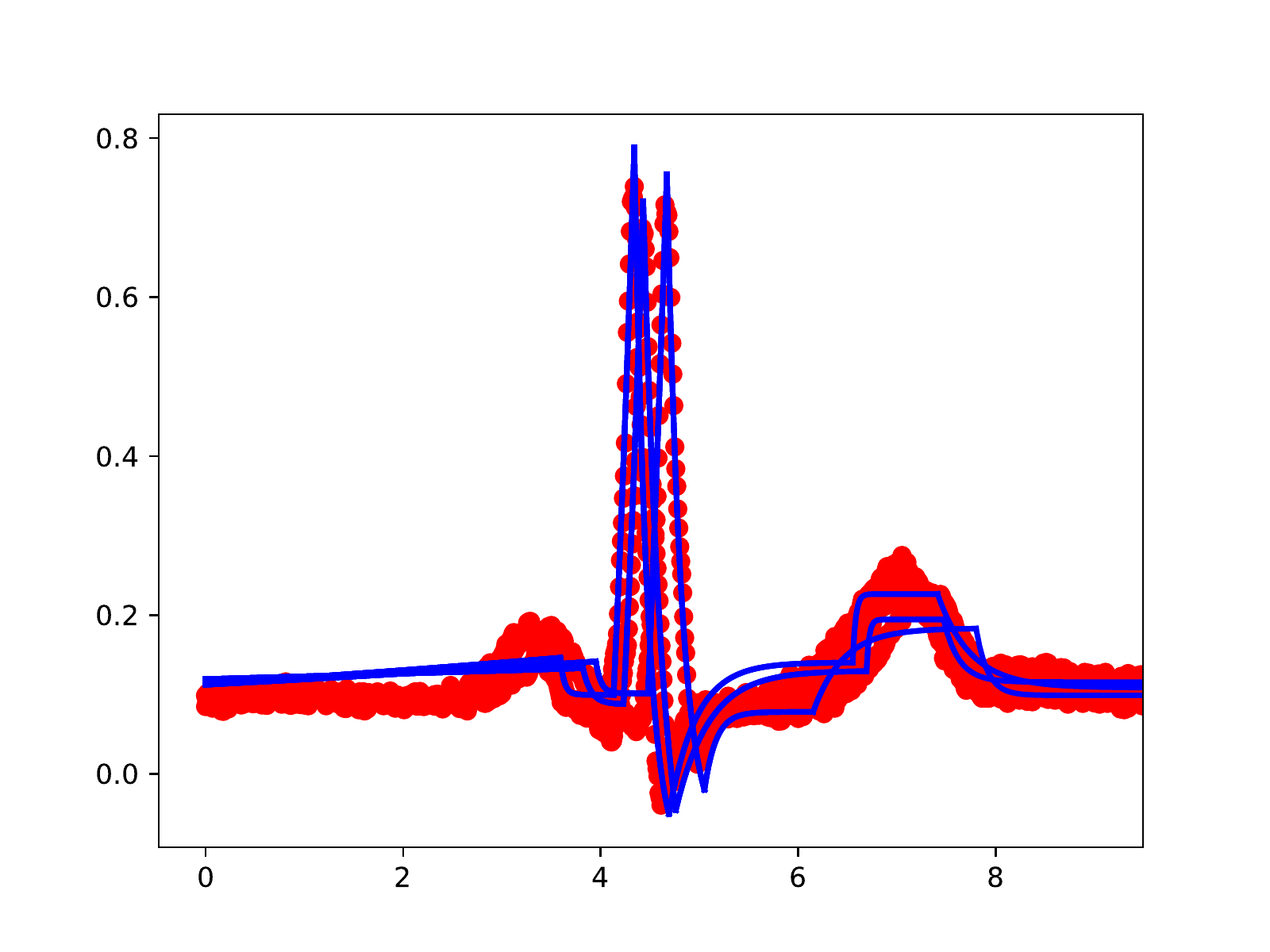}
		\caption{$\delta = 0.05$.}
		\label{fig:time_series_approximation_0_05}
	\end{subfigure}
	\hfill
	\begin{subfigure}{.48\linewidth}
		\centering
		\includegraphics[width=\textwidth,keepaspectratio,clip,trim=12mm 0mm 16mm 0mm]{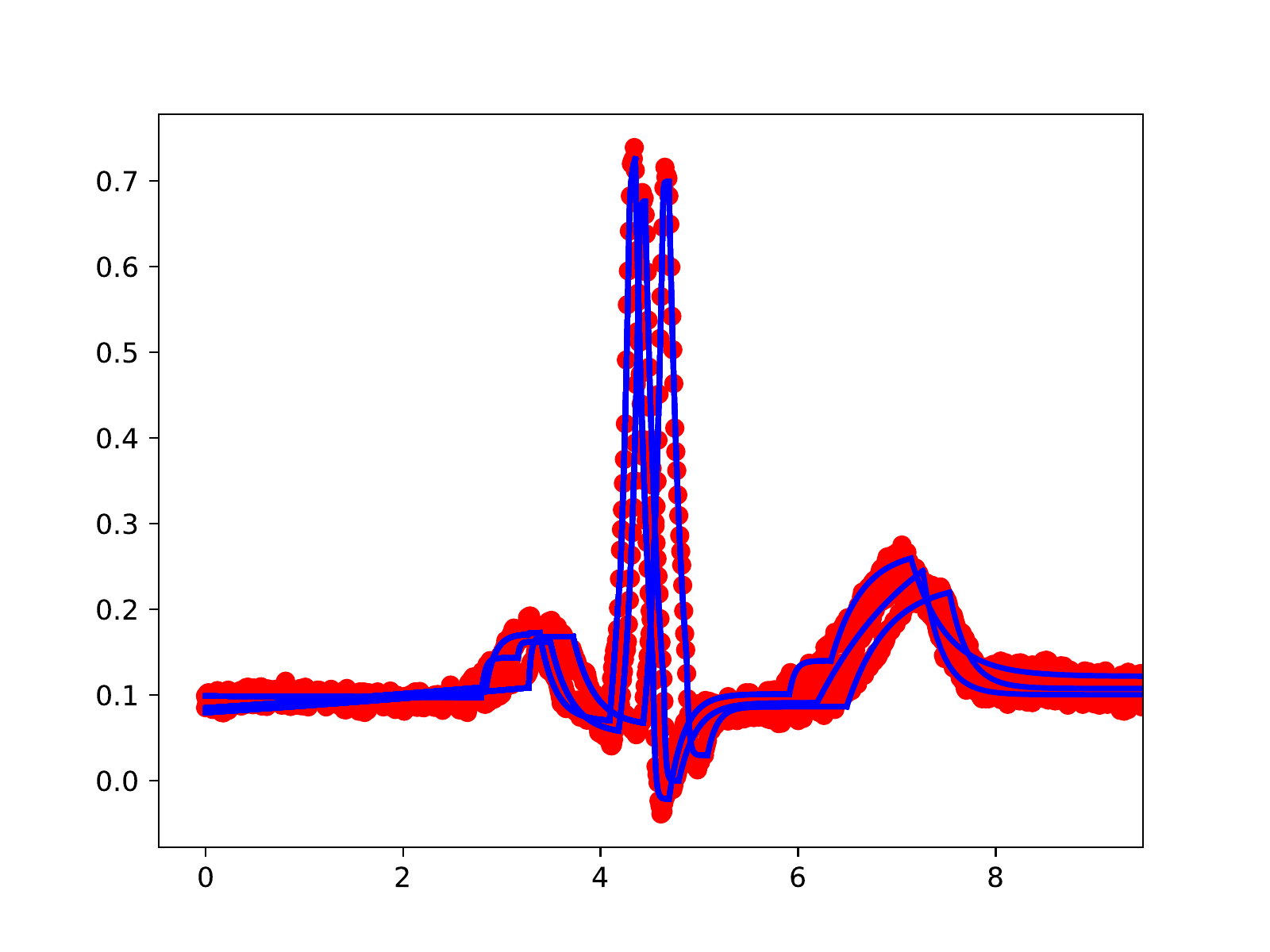}
		\caption{$\delta = 0.02$.}
		\label{fig:time_series_approximation_0_02}
	\end{subfigure}
	
	\begin{subfigure}{.48\linewidth}
		\centering
		\includegraphics[width=\textwidth,keepaspectratio,clip,trim=12mm 0mm 16mm 0mm]{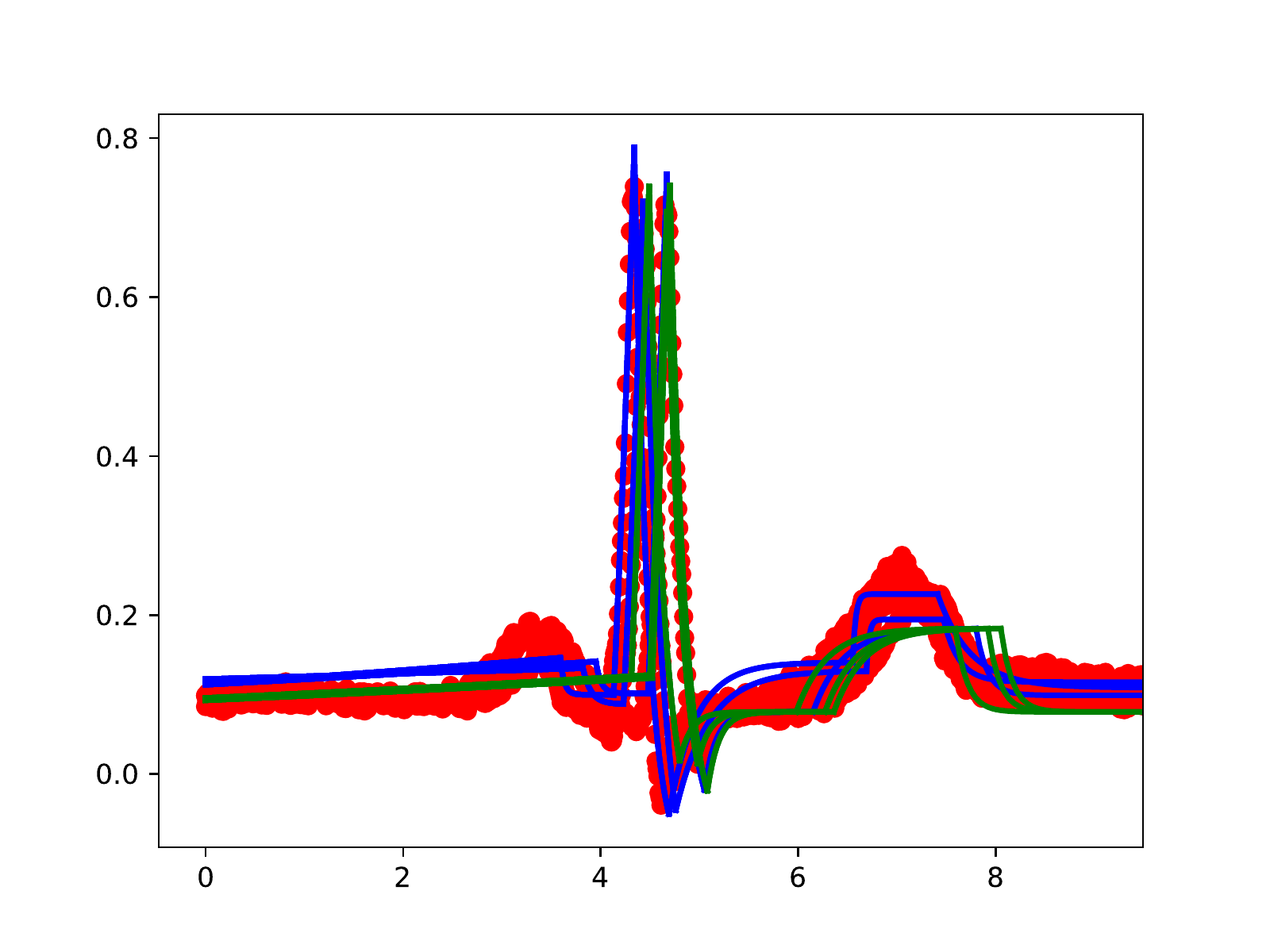}
		\caption{Simulations.}
		\label{fig:ecg_time_series_and_trajectories_0_05}
	\end{subfigure}
	\caption{\pwa trajectories (blue) for three time series (red)
	and different values of $\delta$.
	The last plot shows three simulations from the synthesized
	\adha (green).}
	\label{fig:time_series_approximation}
\end{figure}

In another experiment we investigate the conversion of time series to \pwa
trajectories.
We consider three \textbf{ECG signals} from the PhysioBank
database~\cite{GoldbergerEtal00}.
For the distance value $\delta = 0.05$ we obtained three \pwa trajectories of
length $7$ in 158~seconds.
For $\delta = 0.02$ we obtained \pwa trajectories of
respective lengths $10$, $11$, and $13$ in $220$~seconds.
Using $\eps = 0.1$ we
obtained an \adha with $8$ and $7$ locations, respectively.
Figure~\ref{fig:time_series_approximation} shows the time series, the
\pwa trajectories, and simulations from the synthesized \adha.

We summarize further benchmark results in Table~\ref{tab:results}, where
we also vary the precision parameters (\eps and $\delta$).
As expected, decreasing \eps results in bigger \adha since
existing modes can be shared for different \pwa trajectory pieces less often.
In the ECG benchmark we observe that decreasing $\delta$ can result in smaller
\adha since the constructed \pwa trajectories are less diverse, even though
they have more pieces (up to $13$ pieces ($\delta = 0.02$) compared to $7$
pieces ($\delta = 0.05$)).
The run time is mainly influenced by the depth of the exploration tree and
hence the length of the pieces, but we observe that the algorithm never comes
close to exploring the full tree.

\begin{table}
	\centering
	\begin{tabular}{@{\hspace*{1mm}} c @{\hspace*{2mm}} @{\hspace*{2mm}} c @{\hspace*{2mm}} c @{\hspace*{2mm}} c @{\hspace*{2mm}} c @{\hspace*{2mm}} c @{\hspace*{1mm}}}
		\toprule
		\multirow{2}{*}{Model} & \multirow{2}{*}{$\eps$ / $\delta$} & \multirow{2}{*}{run time} & \multirow{2}{*}{$|Q|$} & \multicolumn{2}{c}{\# exploration-tree nodes} \\
		& & & & explored & maximal \\
		\toprule
		\multirow{4}{*}{Heater} & $\eps = 0.1$\phantom{0} & \phantom{0,0}53~s & \phantom{0}2 & \phantom{000,}607 & \phantom{000,000,}$409{,}940$ \\
		 & $\eps = 0.07$ & \phantom{0,0}51~s & \phantom{0}3 & \phantom{000,}606 & \phantom{000,000,}$505{,}948$ \\
		 & $\eps = 0.04$ & \phantom{0,0}63~s & \phantom{0}5 & \phantom{000,}755 & \phantom{000,00}$6{,}176{,}776$ \\
		 & $\eps = 0.01$ & \phantom{0,}162~s & 13 & \phantom{00}$2{,}798$ & \phantom{000,}$731{,}667{,}684$ \\
		 \midrule
		\multirow{4}{*}{Gearbox} & $\eps = 0.1$\phantom{0} & \phantom{0,0}12~s & \phantom{0}4 & \phantom{000,0}40 & \phantom{000,000,00}$8{,}762$ \\
		 & $\eps = 0.07$ & \phantom{0,0}12~s & \phantom{0}5 & \phantom{000,0}46 & \phantom{000,000,0}$13{,}844$ \\
		 & $\eps = 0.04$ & \phantom{0,0}12~s & \phantom{0}6 & \phantom{000,0}51 & \phantom{000,000,0}$17{,}837$ \\
		 & $\eps = 0.01$ & \phantom{0,0}17~s & 10 & \phantom{000,}109 & \phantom{000,000,0}$80{,}216$ \\
		 \midrule
		\multirow{2}{*}{ECG} & $\delta = 0.05$ & \phantom{0,}157~s & \phantom{0}8 & \phantom{000,}185 & \phantom{000,00}$7{,}716{,}800$ \\
		 & $\delta = 0.02$ & $3{,}115$~s & \phantom{0}7 & $101{,}145$ & $721{,}419{,}383{,}211$ \\
		 \bottomrule
	\end{tabular}
	\caption{Benchmark results.
	The second column shows the allowed error $\eps$ between \pwa trajectories and \adha resp.\ the allowed error $\delta$ between time series and \pwa trajectories.
	The last two columns show the number of explored tree nodes resp.\ the total number of possible tree nodes.}
	\label{tab:results}
\end{table}

\section{Conclusion}

We have presented an automatic synthesis algorithm for computing a hybrid
automaton with affine differential dynamics $\H$ from a set of
time series $S$ respectively from a set of piecewise-affine trajectories $F$.
Given precision parameters $\delta$ and \eps, the main feature of our
algorithm is that every time series $s$ in $S$ is $\delta$-captured by some
trajectory $f$ in $F$ and that $\H$ is guaranteed to \eps-capture every
function $f$ in $F$, that is, $\H$ contains an execution that has distance
at most \eps from $f$.
Another feature of our algorithm is that it works online, meaning that the
functions $f$ are processed sequentially and we only modify the
intermediate automaton models.

For future work, hardness of the membership problem for the class of automata
that we considered is open.
We currently do not know if that problem is decidable, and if so, what
complexity is required to solve it exactly.
Another interesting but challenging extension of our work is to allow for
transition switches not at a single time point but in a whole time interval.

\begin{acks}
This research was supported in part by the Austrian Science Fund (FWF) under grant Z211-N23 (Wittgenstein Award) and the European Union's Horizon 2020 research and innovation programme under the Marie Sk{\l}odowska-Curie grant agreement No. 754411.
\end{acks}

\bibliographystyle{ACM-Reference-Format}
\bibliography{bibliography.bib}

\end{document}